%% file: CJC.tex
\documentclass[a4paper,preprint]{imsart}

\usepackage{natbib}
\bibpunct{(}{)}{,}{a}{}{;}

\usepackage[OT1]{fontenc}
\usepackage{amsmath,amsthm,amssymb}
\usepackage{graphicx}
\usepackage[pdftex,colorlinks]{hyperref}
\usepackage{hypernat}

\startlocaldefs

\usepackage[ruled,vlined]{algorithm2e}
\newtheorem{proposition}{Proposition}
\newtheorem{lemma}{Lemma}
\theoremstyle{remark}
\newtheorem{remark}{Remark}
\newcommand{\argmax}{\mathop{\mathrm{argmax}}}
\newcommand{\argmin}{\mathop{\mathrm{argmin}}}

\endlocaldefs

\begin{document}

\begin{frontmatter}

  \title{Inferring   sparse Gaussian graphical models with latent structure} 
  
  \runtitle{Inferring sparse GGM with latent structure}

 \begin{aug}
   % indicate corresponding author with \corref{}
   \author{\fnms{Christophe} \snm{Ambroise}},
   \author{\fnms{Julien} \snm{Chiquet}} and 
   \author{\fnms{Catherine} \snm{Matias}}
   \ead[label=e1]{christophe.ambroise}
   \ead[label=e2]{julien.chiquet}
   \ead[label=e3]{catherine.matias@genopole.cnrs.fr}
   \ead[label=u]{http://stat.genopole.cnrs.fr}

   \address{Laboratoire Statistique et G\'enome\\
     523, place des Terrasses de l'Agora\\
     91000 \'Evry, FRANCE\\
     \printead{e1,e2,e3}\\
     \printead{u}} 

   \affiliation{CNRS  UMR 8071 \& Universit\'e d'\'Evry}
   \runauthor{C. Ambroise, J. Chiquet and C. Matias}
 \end{aug}

 \begin{abstract}
   Our  concern  is   selecting  the  concentration  matrix's  nonzero
   coefficients   for  a   sparse  Gaussian   graphical  model   in  a
   high-dimensional setting.  This corresponds to estimating the graph
   of conditional  dependencies between the variables.   We describe a
   novel  framework taking  into  account a  latent  structure on  the
   concentration  matrix. This  latent structure  is used  to  drive a
   penalty  matrix  and thus  to  recover  a  graphical model  with  a
   constrained  topology.  Our   method  uses  an  $\ell_1$  penalized
   likelihood  criterion.   Inference  of  the  graph  of  conditional
   dependencies between  the variates and  of the hidden  variables is
   performed   simultaneously   in   an   iterative   \textsc{em}-like
   algorithm.   The  performances  of  our method  is  illustrated  on
   synthetic  as  well as  real  data,  the  latter concerning  breast
   cancer.
 \end{abstract}


% 62H20 Measures of association (correlation, canonical correlation, etc.)
% 62J07 Ridge regression; shrinkage estimators
% 62H30 Classification and discrimination; cluster analysis

% 66JXX Linear inference, regression
% 62Hxx Multivariate analysis
% 62H12 Estimation


 \tableofcontents

\end{frontmatter}

\input{intro}

\input{sec1}

\input{sec2}

\input{sec3}

\bibliography{../../biblio_common}

\appendix 
\input{appendix}

\end{document}

%% file: intro.tex
\section{Introduction}

Estimating the  concentration matrix  (namely the  inverse  of the
covariance matrix) of a  Gaussian vector in a sparse, high-dimensional
setting  has  received  much  attention recently.   Graphical  models
provide  a convenient  setting for  modelling  multivariate dependence
patterns.  In this framework, an undirected graph is matched to the
Gaussian  random   vector,  where  each  vertex   corresponds  to  one
coordinate  of the  vector, and  an edge  is not  present  between two
vertices  if  the  corresponding  random  variables  are  independent,
conditional on the remaining variables.  Now, conditional independence
between  two coordinates  of  the Gaussian  random vector  corresponds
exactly to a  zero entry in the concentration  matrix. Thus, detecting
nonzero   elements  in  the   concentration  matrix   is equivalent  to
reconstructing  the   Gaussian  graphical  model  \citep[GGM,  see
e.g.][]{1996_Book_Lauritzen}.

We focus  here on the  crucial problem of selecting  the concentration
matrix's nonzero  coefficients. In other words, we  focus on variable
selection   rather  than   estimation.    Application  areas   include
gene-regulation  graph  inference in  Biology  (using gene  expression
microarray data),  as well as spectroscopy,  climate studies, functional
magnetic  resonance imaging,  etc. We  provide a  very  novel approach
driving  the  graph  selection  according  to  an  unobserved  modular
structure on the vertices.
\\

The  idea  of covariance  selection  first  appeared  in the  work  of
\citet{1972_Biometrics_Dempster}.  In the  so-called ``large $p$, small
$n$'' setting (namely when the  number of observations is  smaller
than the dimension of the  observed response), the need for covariance
selection is huge, as the empirical covariance matrix is no longer
regular. 

In \citet{2007_SS_drton}, a classification of the different methods for
model  selection/estimation   in  GGMs  into  three   group  types  is
suggested:  constr\-aint-based  methods,  performing statistical  tests;
Bayesian approaches; and score-based methods, maximizing a model-based
criterion.  The  multiple testing problem has been  taken into account
in  \citet{2007_SS_drton,2008_JSPI_drton}.   The  authors  perform
GGM covariance 
selection by  multiple testing  of hypotheses about  vanishing partial
correlation  coefficients.  Such  procedures may  also  be implemented
using  the  PC-algorithm  \citep{2007_JMLR_Kalisch}. Starting  from  a
complete, undirected graph, the PC-algorithm deletes edges recursively,
according to conditional independence decisions.  However, the statistical
procedure in \citet{2007_SS_drton} relies on asymptotic considerations,
a regime never attained in real situations.

Another  attempt in  this vein  is to  consider limited-order partial
correlations      \citep{2006_SAGMB_Wille,2006_JMLR_Castelo}.       In
\citet{2006_SAGMB_Wille}   the   authors   consider  only   zero   and
first-order  conditional  dependencies.  They  argue  that for  sparse
graphical   models,  these  low-order  dependencies   still  reflect
reasonably  well  the  full-order  conditional  dependency  structure.
Moreover, these dependencies  may be well estimated even  with a small
number  of observations.   In  \citet{2006_JMLR_Castelo}, the  authors
introduce a \textit{non-rejection rate} to reduce the multiple testing
and computational problems to which these approaches give rise.
\\

A        Bayesian        framework        is        proposed        in
\citet{2004_JMVA_Dobra,2005_SS_Dobra} and their method was applied for
evaluating  patterns  of association  in  large-scale gene  expression
data. The  approach is based  on \textit{dependency networks},  namely a
collection            of           conditional           distributions
$\{\mathbb{P}(X_i|X_{\backslash   i})\}$  (where   $X_{\backslash  i}$
stands  for  the set  of  all  variables  but $X_i$).   However,  such
conditional  distributions will not  in general  result in  a coherent
joint distribution.  This point  is further discussed below concerning
a similar problem appearing in \citet{2006_AS_Meinshausen}.  Moreover,
constructing  priors  on the set of concentration  matrices is not a
trivial task  (it mainly  relies on Wishart  priors). The use  of MCMC
procedures limits the range of applications to moderate-sized networks.
\\

Before  focusing  on  score-based  methods,  let  us  first  introduce
regularization procedures.   In the context of  linear regression, the
\textsc{Lasso}  (least  absolute  shrinkage  and  selection  operator)
technique   was  introduced  by   \citet{1996_JRSS_Tibshirani}.   This
procedure  performs model  selection and  parameter estimation  at the
same time.  The  idea is that ordinary least-squares criterion may be
improved  in a  sparse context,  using an  $\ell_1$-norm  penalty. The
$\ell_1$-norm penalty  shrinks the estimates to  zero while preserving
the convexity of  the optimization  problem.  Note that  the $\ell_1$-norm
penalization  is also known  as 'basis  pursuit' in  signal processing
\citep{2001_Siam_Chen}.

It is well known that if the ultimate goal is parameter estimation, model
selection  and estimation  should be  done in  a single  step. Performing the
model selection prior to parameter estimation in the
selected  model will, in fact, result in  a  non-robust  procedure. 
However,  our
primary focus here  is on model selection, as we  want to infer sparse
networks. We therefore concentrate on model selection
rather than on estimation performances.

The \textsc{Lars}  algorithm \citep{2004_AS_Efron} is one  of the most
popular techniques  for solving  the \textsc{Lasso}  problem.  It
gives  the  path  of  solutions  obtained  when  varying  the  penalty
parameter (the  penalty parameter is used  as a scaling  factor of the
$\ell_1$-norm penalty). The larger  the penalty parameter, the sparser
the \textsc{Lasso} solution.

Using     convex     optimization     techniques    \citep[see     for
instance][]{1986_book_Minoux},  the   \textsc{Lasso}  problem  may  be
stated as a primal problem,  whose dual formulation may be solved more
easily.   This approach  is taken  in  \citet{2000_JCGS_Osborne}.  The
authors obtain an iterative  algorithm, the ``homotopy method''
\citep{2000_JNA_Osborne}. Other very efficient approaches are based on
focusing on each coordinate iteratively.  Indeed, for each coordinate,
the \textsc{Lasso}  problem is solved very simply  (assuming the other
coordinates are  fixed) by soft-thresholding \citep{1995_JASA_Donoho}.
Thus,  different  'coordinate  optimization'  procedures  have  been
proposed    in    the    literature.     Following   the    work    of
\citet{1998_JCGS_Fu},    a   cyclic    procedure   is    proposed   in
\citet{2007_AAS_Friedman},  where optimization  with  respect to  each
coordinate is done  iteratively; whereas \citet{2008_AAS_Wu} propose a
greedy  approach,  computing  the  solution for  each  coordinate  and
choosing that which provides the  largest decrease in a surrogate objective
function. Note that these approaches rely on the underlying assumption
that the predictors for the regression problem are uncorrelated.
\\

Let us now come back to covariance (or concentration) matrix inference
in GGMs, using maximization of a model-based criterion.

\citet{2006_AS_Meinshausen}   were   the   first  authors   to   apply
\textsc{Lasso} techniques  for inferring a covariance  matrix in a GGM.
Their  approach is  to solve  $p$ different  \textsc{Lasso} regression
problems, where $p$ is the  dimension of the observed vector. The main
drawback of such a procedure is that a symmetrization step is required
to obtain the final network. It might, for instance, be the case that
the estimator of the regression  coefficient for  $X_i$ on $X_j$  is
zero, whereas  the  estimator for $X_j$ on  $X_i$ is not zero. \citeauthor{2006_AS_Meinshausen} propose to use either an ``AND''
or an ``OR'' final step procedure to recover an undirected correlation
graph.   However,  these  two  procedures might  result  in  different
estimates and there is no  way of choosing between them.  Moreover, as
previously  stated,   a  set  of  conditional   distributions  does  not
necessarily cohere into a joint distribution.  Using a set of possibly
non-coherent   conditional    distributions   corresponds    to   a
pseudo-likelihood  approach.   This   aspect  was  not  underlined  in
\citet{2006_AS_Meinshausen},   and    we   clarify   this    point   in
Section~\ref{sec:pseudo-likelihood}.

Subsequently,    two    other    articles,   \citet{2008_JMLR_Banerjee}    and
\citet{2007_Biometrika_Yuan} independently provided  an improvement of
the initial  work of \citet{2006_AS_Meinshausen}.  In  both works, the
problem is seen  as a penalized maximum-likelihood  (PML) problem. Instead
of considering  $p$ different regression problems,  these two articles
focus on the likelihood of the Gaussian vector, penalizing the entries
of  the  concentration matrix  with  an  $\ell_1$-norm penalty.   They
explain   how    the   PML   estimation    may   be   solved    as   a
``\textsc{Lasso}-like'' problem.  The  major issue with PML strategies
in the context of the concentration matrix estimation of GGMs is to obtain
a positive  definite estimate.  However, the approach  for solving the
problem   in   \citet{2007_Biometrika_Yuan}    is   not   suited   to
high-dimensional    settings,  in contrast to the approach proposed in
\citet{2008_JMLR_Banerjee}.     In   \citet{2007_Biometrika_Yuan},   a
non-negative garrote-type estimator  is used, and asymptotic properties
(as $n$ tends to infinity while  $p$ is held fixed) are given.  In \citet{2008_JMLR_Banerjee}, two  different algorithms
are proposed for solving problems  in a high-dimensional setting.  The
first  approach relies  on a  block-coordinate descent  algorithm. The
second   is  a   semi-definite  programming  algorithm,  based  on
Nesterov's method, which is computationally intensive.

The   next    improvement   in   this   vein   comes with
\citet{2007_BS_Friedman}.   Relying on coordinate  descent techniques,
previously described  in \citet{2007_AAS_Friedman}, the  authors revisit
\citeauthor{2008_JMLR_Banerjee}'s   first  approach  and   propose  an
efficient  algorithm to solve  the PML  estimation problem,  under the
positive definite  constraint. In fact, they use  the block coordinate
descent  approach proposed  by \citet{2008_JMLR_Banerjee} and combine it
with a second coordinate descent  method.  Our method will make use of
this approach.

To  conclude this  part, we remark  that a  completely  different shrinkage
estimate  was  proposed  by  \citet{2005_SAGMB_Schafer}  in  the  same
context  of  large-scale covariance  matrix  estimation. The  approach
consists in using a weighted  average of two different estimators, the
first being unconstrained (thus  having small bias but large variance),
the second being low-dimensional (and thus exhibiting small variance
but large bias).
\\

Now  let  us  motivate  the  use of  hidden  structures  in  networks.
Modularity  is  a  property  observed in  real  (biological)  networks
\citep[see  for  instance][]{2002_NatureGen_Ihmels}. Heterogeneity  in
the  node behaviors  is  an  important property  of  these data.   For
example,  so-called 'hubs'  are highly-connected nodes,  showing a
different behavior from  the rest of the graph  nodes.  An interesting
model capturing these network features is a mixture model for random
graphs \citep[see for instance][]{2006_SC_Daudin}. This model has been
rediscovered  many  times in  the  literature,  and  a non-exhaustive
bibliography   should   include   \citet{FH82,SN97,   NS01,Tallberg05,
  2006_SC_Daudin,  2007_report_Mariadassou,  2008_report_Zanghi}.   To
state it simply, this model assumes that each
node belongs to some unobserved group. Conditional on the node groups,
the  (weighted)  edges  are  independent and  identically-distributed
(i.i.d.)  random  variables, whose distribution depends  on the groups
of the nodes  to be connected. As we are  interested in GGMs, weighted
edges correspond to entries of the concentration matrix.
\\

In this  work, we  aim at estimating  a hidden structure,  namely node
groups, while  discovering the  network. This hidden  structure should
help us  in choosing \textit{adaptive} penalty  parameters. Indeed, we
wish to  penalize the elements of the  concentration matrix, according
to the unobserved clusters to which the nodes belong. For instance, if two nodes
belong  to the same  unobserved group,  we wish  to lower  the penalty
parameter  acting  on the  corresponding  entry  in the  concentration
matrix. Conversely,  if we increase the  penalty parameters  on the
entries  corresponding  to nodes  belonging  to  different groups,  we
shrink the  estimated coefficient to zero. Our  approach is completely
new and improves inference of sparse modular networks.
\\

Another    adaptive    \textsc{Lasso}    procedure   is    given    in
\citet{2006_JASA_Zou}, whose idea is to lower the bias of the
large  coefficients  by  adapting   the  penalty  parameter  of  each
coefficient  so  that  it  automatically scales  with  the  inferred
value. It is known that  the non-adaptive \textsc{Lasso} procedure may
result  in  inconsistent  parameter  estimation. An illustration  of  the
conflict between optimal  prediction and consistent variable selection
for     the      \textsc{Lasso}     procedure     is      given     in
\citet{2006_AS_Meinshausen}.   They proved  that  the optimal  penalty
parameter  for   prediction  gives  inconsistent   variable  selection
results, motivating the  use of another  penalty parameter
to ensure the control  of the probability of falsely  connecting two or
more  distinct connectivity components  of the  graph.  Like them, we also focus
on  optimal selection rather than on optimal  prediction. The adaptivity of
our procedure is not used for lowering the bias of large coefficients,
but  instead for constraining  the  prediction  to  fit  the  underlying
structure of the graph.

\paragraph{Model.}  Let us now briefly describe the general approach of
our   work.   The   model   will   be  presented in detail   in
Section~\ref{sec:framework}.  Let $X = (X_1,\dots,X_p)^\intercal$ be a
Gaussian random vector in  $\mathbb{R}^p$, with zero mean and positive
definite  covariance matrix  $\mathbf{\Sigma}$, namely  $  X\sim {\cal
 N}(\mathbf{0}_p,\mathbf{\Sigma}).   $  We  observe  independent  and
identically-distributed (i.i.d)  vectors  $(X^1,\dots,X^n)$ with  the
same    distribution    as   $X$.     The    matrix   $\mathbf{K}    =
\mathbf{\Sigma}^{-1}$ is  the concentration matrix of  the model.  Let
$\mathbf{S}$ be  the empirical covariance  matrix.  The log-likelihood
of the observations is given by
\begin{equation*}
  \mathcal{L}(\mathbf{K} ) 
  = \frac{n}{2} \log
  \det(\mathbf{K}) - \frac{1}{2} \sum_{k=1}^n (X^k)^\intercal \mathbf{K} X^k +c 
  = \frac{n}{2} \log  \det(\mathbf{K}) -  \frac{n}{2} \text{Tr}(\mathbf{S}
  \mathbf{K}) +c
\end{equation*}
where $c$ is a constant term.

The $\ell_1$-penalized estimator proposed by \citet{2008_JMLR_Banerjee}
is given by
\begin{equation}
  \label{eq:penal_baner}
  \widehat{\mathbf{K}} = \arg \max_{\mathbf{K} \succ 0} \; \log
  \det(\mathbf{K}) - \text{Tr} (\mathbf{S}\mathbf{K}) - \rho \|\mathbf{K}\|_{\ell_1} ,
\end{equation}
where $\mathbf{K} \succ 0$  stands for positive definiteness, $\rho>0$
is  a  penalty  parameter  and  $\|\mathbf{K}\|_{\ell_1}  =  \sum_{ij}
|K_{ij}|$.

A natural generalization of this approach is to have different penalty
parameters for different entries $K_{ij}$. Namely,
\begin{equation*}
  \log     \det(\mathbf{K})     -     \text{Tr}    (\mathbf{S}\mathbf{K})     -
  \|\boldsymbol{\rho}(\mathbf{K})\|_{\ell_1},
\end{equation*}
where    $\boldsymbol{\rho}(\mathbf{K})=(\rho_{ij}(K_{ij}))_{   i,j\in
 \mathcal{P}}$  is  a matrix  of  penalty  functions  acting on  each
entry.  As a general rule, using  as many penalty functions as there are entries in the
concentration matrix to be estimated is not meaningful.

Here,  we propose  to  take into  account  a hidden  structure on  the
correlations between the coordinates  random variables $X_k$. Thus, we
consider latent  i.i.d.  random variables  $\mathbf Z_1,\ldots,\mathbf
Z_p$  with values  in a  finite set  $\{1,\ldots,Q\}$.   Each variable
$\mathbf Z_i$  describes the  \textit{state} of $X_i$,  and we  wish to
adapt the penalty  function $\rho_{ij}$ with respect to  the states of
$X_i, X_j$. More precisely, we wish to use a criterion of the form
\begin{equation*}
  \log  \det(\mathbf{K}) - \text{Tr} (\mathbf{S}\mathbf{K}) -  \|\boldsymbol{\rho}_{\mathbf Z}(\mathbf{K})\|_{\ell_1} ,
\end{equation*}
where   $\boldsymbol{\rho}_{\mathbf{Z}}(\mathbf{K})=(\rho_{\mathbf{Z}_i
  \mathbf{Z}_j}(K_{ij}))_{ i,j\in \mathcal{P}}$  is a matrix of random
penalty  functions  whose  entries  depend  on  the  latent  structure
$\mathbf  Z=\mathbf  Z_1,\ldots, \mathbf  Z_p$.   However, the  hidden
structure is not supposed to be  known, thus we cannot rely on the
previous  criteria.   Intuitively,  following  the principle  of
Expectation-Maximization        (\textsc{em})       algorithm       of
\citet{1977_JRSSB_Dempster}, the idea will be to  replace  the   unobserved  value
$\boldsymbol{\rho}_{\mathbf{Z}}(\mathbf{K})$   with   its  conditional
expectation                  $\mathbb{E}(\boldsymbol{\rho}_{\mathbf{Z}}
(\mathbf{K})|X^1, \ldots,X^n; \mathbf{K}^{(m)})$ under some model with
parameter $\mathbf{K}^{(m)}$, and iterate the following steps
\begin{itemize}
\item[(\textsc{e})] Compute $\mathbb{E}(\boldsymbol{\rho}_{\mathbf{Z}}
  (\mathbf{K})|X^1, \ldots,X^n; \mathbf{K}^{(m)})$
\item[(\textsc{m})]  Update $\mathbf{K}^{(m+1)}  = \argmax_{\mathbf{K}
    \succ      0}      \mathbb{E}      (\boldsymbol{\rho}_{\mathbf{Z}}
  (\mathbf{K})|X^1, \ldots,X^n; \mathbf{K}^{(m)})$.
\end{itemize}

One of our  aims is to provide  a very simple framework for such an
analysis. 
%Full details are given below.

Note that the $\ell_1$-norm used here acts on diagonal elements of the
matrix  $\mathbf  K$. It  is  counter-intuitive  to penalize  diagonal
elements of the concentration matrix, as these do not reflect sparsity
in the correlation structure. However, from a technical point of view,
this strategy  ensures  that  the  procedure  will  select  a  positive  definite
estimator  (see   Remark~\ref{rem:pen_diag}).   This  point   was  not
emphasized  in  the   previous  procedures  using  $\ell_1$  penalized
likelihood  of  GGMs.

\paragraph{Road-map.}  
In Section~\ref{sec:framework} we present the model and the penalized
maximum-likelihood criterion on which we base our inference procedure,
described in Section~\ref{sec:inference}.  This procedure relies on
a variational \textsc{em} algorithm, combined with a \textsc{Lasso}-like
procedure.  Section~\ref{sec:pseudo-likelihood} explains how
\citeauthor{2006_AS_Meinshausen}'s approach may be interpreted as a
penalized pseudo-likelihood method.  Finally, Section~\ref{sec:numexp}
illustrates the performance of the method on synthetic data, for which
an \verb:R:--package, \texttt{SIMoNe} (Statistical Inference for
Modular Network), can be downloaded from the first author's website.  We
also test our algorithm on a real data set provided by
\citet{2006_JCO_hess} and concerning $n=133$ patients with breast
cancer treated using chemotherapy.  According to \citet{2006_JCO_hess}
and \citet{2008_BMC_Natowicz}, the patient response to the treatment can
be classified either as a pathologic complete response (pCR), or
as a residual disease (not-pCR). The prediction of the patient response is
achieved accurately by studying the expression levels of a limited
number of genes ($p=26$).  Our algorithm is applied on each class of
patients (pcR and not-pCR).  Two distinct gene-regulatory networks are
thus inferred, showing a very different structure according to the
selected class of patients.

%%% Local Variables:
%%% mode: latex
%%% TeX-master: "CJC.tex"
%%% End:

%% file: sec1.tex
\section{A latent structure model for network inference}\label{sec:framework}

In this  section we present  a framework for  modelling heterogeneity
among dependencies  between the variables.   To this end,  let us
first  recall  classical  notations  from  Gaussian  Graphical  Models
\citep[see][for elementary results about GGMs]{1996_Book_Lauritzen}.

\subsection{Gaussian graphical models: general settings}

Let   $\mathcal{P}=\{1,\dots,p\}$  be   a  set   of   fixed  vertices,
$X=(X_1,\dots,X_p)^\intercal$ a random vector describing a signal over
this  set and a  sample $(X^1,\dots,X^n)$  of size  $n$ with  the same
distribution as  $X$. 

The  vector $X$  is  assumed  to be  Gaussian  with positive  definite
covariance           matrix           ${\boldsymbol\Sigma}           =
(\Sigma_{ij})_{(i,j)\in\mathcal{P}^2}$.   No  loss  of  generality  is
involved    when   centering    $X$,    so   we    may   assume    that
$X\sim\mathcal{N}(\mathbf{0}_p,{\boldsymbol\Sigma})$.   GGMs are based
on     a    classical     result,     originally    emphasized     by
\citet{1972_Biometrics_Dempster}, claiming that  any couple of entries
$(X_i,X_j)$ with  $i\neq j$ are  independent conditional on  all other
variables indexed by $\mathcal{P} \backslash \{ i,j \}$, if and only if the
entry $({\boldsymbol\Sigma}^{-1})_{ij}$  is zero.  The  inverse of the
covariance  matrix   $\mathbf{K}  =  (K_{ij})_{(i,j)\in\mathcal{P}^2}=
{\boldsymbol\Sigma}^{-1}$,  known as  the  concentration matrix,  thus
describes  the conditional  independence structure  of  $X$. Moreover,
each  entry  $K_{ij}, i\neq  j$  is  directly  linked to  the  partial
correlation  coefficient $r_{ij|\mathcal{P}\backslash\{i,j\}}$ between
variables     $X_i$     and    $X_j$.      In     fact,    we     have
$r_{ij|\mathcal{P}\backslash\{i,j\}} = -K_{ij} / \sqrt{K_{ii} K_{jj}}$,
and also $K_{ii}= \textrm{Var}(X_i|X_{\mathcal{P}\backslash i})^{-1}$.
Hence,  after  a simple  rescaling,  the  matrix  $\mathbf{K}$ can  be
interpreted as  the adjacency matrix  of an undirected  weighted graph
$\mathcal{G}$ representing  the partial correlation  structure between
variables  $X_1,\ldots,X_p$.   This graph  has  no  self-loop, with  a
random set  of edges composed by  all pairs $(i,j)$  such that $K_{ij}
\neq 0$.   Note that we are seeking only pairs of  vertices $(i,j)$ such
that  $i  < j$,  since  there is  no  self-loop,  and since  $K_{ij}=K_{ji}$.
Inferring  nonzero entries  of  $\mathbf{K}$ is  equivalent to  inferring
$\mathcal{G}$,  and is therefore a highly relevant issue  in this
framework.

\subsection{Providing the network with a latent structure}

Let  us  now extend  the  modeling by  providing  the  network with  an 
internal latent structure.  

The model  proposed in \cite{2006_SC_Daudin} attempts a  better fit of
data, as it places the network $\mathcal{G}$ in the mixture framework,
in order  to take account of the  heterogeneity among vertices. The
same general mixture model  is adopted here: vertices of $\mathcal{P}$
are  distributed among  a set  $\mathcal{Q}=\{1,\dots,Q \}$  of hidden
clusters  that model  the latent  structure of  the network.   For any
vertex  $i$, the  indicator  variable $Z_{iq}$
%)_{q\in\mathcal{Q}}$  
is equal to  $1$ if  $i\in q$ and  $0$ otherwise, hence  describing which
cluster the vertex  $i$ belongs to.  A vertex is  assumed to belong to
one  cluster only,  thus the  random vector  $\mathbf{Z}_i  = (Z_{i1},
\dots, Z_{iQ})$ obviously follows a multinomial distribution. Namely, 
\begin{equation}
  \label{eq:distribution_Zi}
  \mathbf{Z}_i \sim \mathcal{M}(1,\boldsymbol \alpha),
\end{equation}
where $\boldsymbol \alpha =  (\alpha_1,\dots,\alpha_Q)$ is a vector of
cluster proportions, so that $\sum_q \alpha_q=1$.

\paragraph{The concentration  matrix structure.}  We  shall now extend
the clustering of vertices from  $\mathcal{P}$  to the  concentration  matrix
$\mathbf{K}$.  Accordingly, both the existence and  the weight  of edges,
described by the off-diagonal elements of $\mathbf{K}$, will depend on
the cluster each vertex belongs to. Conditional on the events $i\in q$
and $j\in \ell$ where $q,\ell$ are clusters chosen from $\mathcal{Q}$,
each  $K_{ij}$ ($i\neq  j$)  is a  random  variable whose  probability
density function is denoted by $f_{q\ell}$, that is,
\begin{equation}
  \label{eq:conditionnal_distribution_Wij}
  K_{ij} | \left\{Z_{iq} Z_{j\ell} = 1\right\} \sim f_{q\ell}(\cdot), \quad i\neq j.
\end{equation}
It will be remarked that in this formulation  the variables $K_{ij}$ are assumed to
be  independent, conditional on  the clusters  the vertices  belong to.
Moreover,  we  are considering  only  undirected  graphs,  so  we  may assume that
$f_{q\ell}=f_{\ell     q}$.      For     technical    reasons     (see
Remark~\ref{rem:pen_diag}),  we also  assume a  prior  distribution on
diagonal elements of $\mathbf K$, namely
\begin{equation*}
  K_{ii}\sim f_0(\cdot).
\end{equation*}

Our suggestion is to adopt Laplace distributions; hence
\begin{equation}
  \label{eq:fql_Laplace}
  \forall x \in \mathbb{R} , \quad 
  f_{q\ell}(x)               =              \frac{1}{2\lambda_{q\ell}}
  \exp\left\{-\frac{|x|}{\lambda_{q\ell}}\right\}, \quad \text{and} \quad f_0(x)               =              \frac{1}{2\lambda_0}
  \exp\left\{-\frac{|x|}{\lambda_0}\right\}, \
\end{equation}
where  $\lambda_{q\ell}, \lambda_0  >0$ are  scaling  parameters and $\lambda_{q\ell}=\lambda_{\ell q}$.  Below, the
parameter  $\lambda_0$  will  be   fixed  and  not  estimated.

The reason for choosing a  Laplace  distribution is
that it is reminiscent of the $\ell_1$-norm,  itself linked to
\textsc{Lasso}-techniques  for  which appropriate  tools  are available.   In
fact, when considering the general penalized least-square problem, the
penalty  term can  be seen  as a  log-prior density  on the  vector of
parameters. In the case  of \textsc{Lasso}, the prior distribution
corresponding   to   the  $\ell_1$-norm   is   actually  the   Laplace
distribution \citep[see, e.g.][]{2001_book_Hastie}.

\paragraph*{The affiliation model.} The affiliation model is a special
case  of network  structure (to be investigated below), where there are many
different  clusters, but where the focus is restricted to  two types of edges:
edges between nodes  of the same cluster, and  edges between nodes from
different  clusters.   In  the affiliation model the  densities  $f_{q\ell}$  in
\eqref{eq:fql_Laplace} are  of only two  kinds;  that is, for
all $q,\ell\in\mathcal{Q}$, let
\begin{equation}\label{eq:affiliation}
  f_{q\ell} = \left\{ 
  \begin{array}{lll}
    f_{qq} =  f_{\text{in}}(\cdot;\lambda_{\text{in}}) & \text{if }
    q=\ell, & \text{the \emph{intra-cluster} density of edges}, \\
    f_{q\ell}    =    f_{\text{out}}(\cdot;\lambda_{\text{out}})   &
    \text{if } q\neq \ell, & \text{the \emph{inter-cluster} density of edges}.
  \end{array} \right.
\end{equation}

\subsection{The complete likelihood}

Having described the modeling of the network, we now focus on the
inference issue.

We denote  as  $\mathbf{X}$ the  $n\times  p$  matrix  that contains  the
data-set    $\{X^1,X^2,\dots,X^n\}$   row-wisely    organized,   i.e.,
$(X^k)^\intercal$  is  the  $kth$  row of  $\mathbf{X}$.   
Furthermore, we denote   as    $\mathbf{Z}    =   \{Z_{iq}\}_{i\in\mathcal{P},
 q\in\mathcal{Q}}$  the set  of  all latent  indicator variables  for
vertices. For the  sake of simplicity, the number  of clusters $Q$ and
the parameters  $\boldsymbol \alpha=(\alpha_q)_{q\in \mathcal{Q}}$ and
$ \boldsymbol \lambda  = \{\lambda_{q\ell}\}_{q,\ell \in \mathcal{Q}}$
are assumed to be known for the moment.

The data experiments $\mathbf{X}$ are the only observations available,
and from these we should like to be able to infer  the graph $\mathcal{G}$
of  conditional  dependencies or,  equivalently,  nonzero entries  of
$\mathbf{K}$.   As the  matrix  $\mathbf{K}$ has  been  given a  prior
distribution,  our  aim  is to maximize the  posterior  probability  of
$\mathbf{K}$,  given  the   data  $\mathbf{X}$,  or  equivalently,  the
logarithm of the joint distribution.   The estimate is thus defined as
follows:
\begin{equation*}
  \label{eq:posterior_to_joint}
  \widehat{\mathbf{K}} = \arg \max_{\mathbf{K}\succ 0} \mathbb{P}(\mathbf{K}
  | \mathbf{X}) = \arg \max_{\mathbf{K}\succ 0} \ \log
  \mathbb{P}(\mathbf{X},\mathbf{K}), 
\end{equation*}
where $\mathbf{K} \succ 0$ stands for positive-definiteness.

To  solve   this  problem,  we   place  ourselves  in   the  classical
complete-data  framework. The  distribution of
$\mathbf{K}$  is  only  known  conditionally  on  the  latent  structure
described  by $\mathbf{Z}$.  We denote  as $\mathcal{Z}$  the set  of all
possible clusterings over nodes from $\mathcal{P}$. The marginalization
over the latent clusters $\mathbf{Z}$ leads to
\begin{equation*}
  \label{eq:joint_lattent}
  \widehat{\mathbf{K}} = \arg\max_{\mathbf{K}\succ 0} \ \log \sum_{\mathbf{Z}\in\mathcal{Z}}
  \mathcal{L}_c(\mathbf{X},\mathbf{K},\mathbf{Z}),
\end{equation*}
where        the       so-called        complete-data       likelihood
$\mathcal{L}_c(\mathbf{X},       \mathbf{K}       ,\mathbf{Z})       =
\mathbb{P}(\mathbf{X},  \mathbf{K}, \mathbf{Z})$  is  the function  we
shall develop using an \textsc{em}-like strategy  hereafter.  For this
purpose, a closed form of $\mathcal{L}_c$ is required.

\begin{proposition}\label{prop:complete_likelihood}The        following
  relation holds for the complete-data likelihood $\mathcal{L}_c$.
  \begin{multline}
    \label{eq:complete_likelihood}
    \log   \mathcal{L}_c(\mathbf{X},\mathbf{K},\mathbf{Z})  = 
    \frac{n}{2}\left(\log          \det         (\mathbf{K})         -
      \mathrm{Tr}(\mathbf{S}\mathbf{K})\right) -\left\|\boldsymbol
      \rho_{\mathbf{Z}}(\mathbf{K}) \right\|_{\ell_1} \\
    -     \sum_{\substack{i,j\in    \mathcal{P} , i\neq j     \\
        q,\ell\in\mathcal{Q}}} Z_{iq}  Z_{j\ell} \log (2 \lambda_{q\ell})
    + \sum_{i\in \mathcal{P}, q\in\mathcal{Q}} Z_{iq} \log \alpha_q + c,
  \end{multline}
  where  $\mathbf{S} =  n^{-1} (\mathbf{X}-\bar{\mathbf{X}})^\intercal
  (\mathbf{X}-\bar{\mathbf{X}})$  is the empirical  covariance matrix,
  $c$ is a  constant term and $\boldsymbol \rho_\mathbf{Z}(\mathbf{K})
  =           \left(\rho_{\mathbf{Z}_i           \mathbf{Z}_j}(K_{ij})
  \right)_{i,j\in\mathcal{P}}$ is defined by
  \begin{equation}
    \label{eq:rho_func}
    \rho_{\mathbf{Z}_i \mathbf{Z}_j}(K_{ij}) = 
    \left\{\begin{array}{lr}
      \displaystyle \sum_{q,\ell\in\mathcal{Q}} Z_{iq}Z_{j\ell}
      \frac{|K_{ij}|}{\lambda_{q\ell}} & \text{if } i\neq j, \\[5ex]
      \displaystyle \frac{|K_{ii}|}{\lambda_{0}} & \text{otherwise}.
    \end{array}\right.
  \end{equation}
\end{proposition}
\begin{proof} Using the Bayes  rule, $\mathcal{L}_c$ divides into three
 terms:
 \begin{equation*}
   \log\mathcal{L}_c(\mathbf{X},\mathbf{K},\mathbf{Z})  = \log
   \mathbb{P}(\mathbf{X},\mathbf{K},\mathbf{Z}) \\ 
   = \log\mathbb{P}(\mathbf{X}|\mathbf{K})       +      \log
   \mathbb{P}(\mathbf{K}|\mathbf{Z}) + \log \mathbb{P}(\mathbf{Z}),
 \end{equation*}
 where    we   make use    of    the    fact   that    $\log
 \mathbb{P}(\mathbf{X}|\mathbf{K},\mathbf{Z})          =         \log
 \mathbb{P}(\mathbf{X}|\mathbf{K})$.\\

 The first term is the likelihood associated with a size-$n$ sample of
 a      multivariate     Gaussian     distribution,      since     $X
 \sim\mathcal{N}(\mathbf{0}_p,{\boldsymbol\Sigma})$.           Routine
 computations lead to
 \begin{equation*}
   \log\mathbb{P}(\mathbf{X}|\mathbf{K}) = \frac{n}{2}\log
   \det   (\mathbf{K}) - \frac{n}{2}\mathrm{Tr}(\mathbf{S} \mathbf{K}) -
   \frac{np}{2}\log (2\pi).
 \end{equation*}
 As regards  the second  term, using the expression  \eqref{eq:fql_Laplace}, we have
 \begin{multline*}
   \log\mathbb{P}(\mathbf{K}|\mathbf{Z}) =\sum_{\substack{i,j\in\mathcal{P} , i\neq j \\     
  q,\ell\in\mathcal{Q}}} Z_{iq} Z_{j\ell} \log f_{q\ell}(K_{ij})
   + \sum_{i \in \mathcal{P}}  \log    f_0(K_{ii}) \\
   =  - \sum_{\substack{i,j\in\mathcal{P} , i\neq j \\     
  q,\ell\in\mathcal{Q}}} Z_{iq} Z_{j\ell}
   \left(   \frac{{   |    K_{ij}|}}{\lambda_{q\ell}}   +   \log   (2
     \lambda_{q\ell})    \right)   -    \sum_{i    \in   \mathcal{P}}
   \frac{|K_{ii}|}{\lambda_0} - p\log (2 \lambda_{0}) .
 \end{multline*}
 From       \eqref{eq:distribution_Zi},      we       have
 $\log\mathbb{P}(\mathbf{Z}) = \sum_{i,q}  Z_{iq} \log \alpha_q$, and
 the result follows.
\end{proof}

%%% Local Variables:
%%% mode: latex
%%% TeX-master: "CJC.tex"
%%% End:

%% file: sec2.tex
\section{Inference strategy by alternate optimization}\label{sec:inference}

In    the    classical     \textsc{em}    framework    developed    by
\cite{1977_JRSSB_Dempster}, where  $\mathbf{X}$ is the  available data,
inferring the unknown  parameters $\mathbf{K}$ spread over a latent
structure  $\mathbf{Z}$ would  make use  of the  following conditional
expectation:
\begin{multline}
  \label{eq:conditional_expectation}
  Q\left(\mathbf{K}|\mathbf{K}^{(m)}\right)                           =
  \mathbb{E}        \left\{       \log
    \mathcal{L}_c(\mathbf{X},\mathbf{K},\mathbf{Z}) \big| \mathbf{X};
    \mathbf{K}^{(m)}\right\} \\
  =                                    \sum_{\mathbf{Z}\in\mathcal{Z}}
  \mathbb{P}\left(\mathbf{Z} \big|\mathbf{X},\mathbf{K}^{(m)}\right) \log
  \mathcal{L}_c(\mathbf{X},\mathbf{K},\mathbf{Z}) 
  =                                     \sum_{\mathbf{Z}\in\mathcal{Z}}
  \mathbb{P}\left(\mathbf{Z} \big|\mathbf{K}^{(m)}\right) \log 
  \mathcal{L}_c(\mathbf{X},\mathbf{K},\mathbf{Z}),
\end{multline}
where $\mathbf{K}^{(m)}$  is the  estimation of $\mathbf{K}$  from the
previous step of the algorithm.
\\

The  usual  \textsc{em} strategy  would  be to alternate an  \textsc{E}-step
computing            the            conditional            expectation
\eqref{eq:conditional_expectation}  with  an \textsc{M}-step  maximizing
this   quantity   over  the   parameter   of  interest   $\mathbf{K}$.
Unfortunately,              no              closed form             of
$Q\left(\mathbf{K}|\mathbf{K}^{(m)}\right)$  can be formulated  in the
present case. The technical  difficulty lies in the complex dependency
structure    contained    in    the    model.     Indeed,    $\mathbb{P}
(\mathbf{Z}|\mathbf{K})$   cannot   be   factorized,  as   argued   in
\cite{2006_SC_Daudin}.    This  makes   the   direct  calculation   of
$Q\left(\mathbf{K}|\mathbf{K}^{(m)}\right)$  impossible.   To tackle
this  problem we  use a  variational approach  \citep[see, e.g.,][for
elementary  results on  variational  methods]{2000_Book_Jaakkola}.  In
this framework,  the conditional distribution of  the latent variables
$\mathbb{P} (\mathbf{Z}|\mathbf{K}^{(m)})$  is approximated by  a more
convenient  distribution  denoted  by  $R_{m}(\mathbf{Z})$,  which  is
chosen   carefully   in   order   to   be   tractable.    Hence,   our
\textsc{em}-like algorithm  deals with the  following approximation of
the conditional expectation \eqref{eq:conditional_expectation}
\begin{equation}
  \label{eq:conditional_expectation_approx}
  \mathbb{E}_{R_m}\left\{\log\mathcal{L}_c(\mathbf{X},\mathbf{K},\mathbf{Z}) 
  \right\}  =  \sum_{\mathbf{Z}\in\mathcal{Z}}  R_m(\mathbf{Z} )  \log
  \mathcal{L}_c(\mathbf{X},\mathbf{K},\mathbf{Z}).
\end{equation}

In the following  section we develop a variational  argument in order
to     choose    an     approximation     $R_m(\mathbf{Z})    $     of
$\mathbb{P}(\mathbf{Z}|\mathbf{K}^{(m)})$. This enables us to compute the
conditional  expectation \eqref{eq:conditional_expectation_approx} and
proceed to the maximization step.

\subsection[Variational       estimation      of       the      latent
structure]{Variational    estimation   of    the    latent   structure
  (\textsc{E}-step)}\label{sec:Estep}

In this part, $\mathbf{K}$ is assumed  to be known, and we are looking
for an  approximate distribution  $R(\cdot)$ of the  latent variables.
The  variational  approach  consists   in  maximizing  a  lower  bound
$\mathcal{J}$  of the  log-likelihood $\log \mathbb{P}(\mathbf{X},\mathbf{K})$,
defined as follows:
\begin{equation}
  \label{eq:definition_J}
  \mathcal{J}\left(\mathbf{X},\mathbf{K},R(\mathbf{Z})    \right)    =
  \log\mathbb{P}(\mathbf{X},\mathbf{K}) 
  -   \mathrm{D}_{KL}\left\{R(\mathbf{Z})  \|   \mathbb{P}(\mathbf{Z}  |
    \mathbf{K})\right\} 
\end{equation}
where  $\mathrm{D}_{KL}$  is  the  K\"ullback-Leibler  divergence.  This
measures   the  difference   between   the  probability   distribution
$\mathbb{P}(\cdot|\mathbf{K})$  in   the  underlying  model   and  its
approximation $R(\cdot)$.  An intuitively straightforward choice for $R(\cdot)$
is         a         completely        factorized         distribution
\citep[see][]{2007_report_Mariadassou,2008_report_Zanghi}
\begin{equation}
  \label{eq:R_factorized}
  R_{\boldsymbol\tau}(\mathbf{Z}) = \prod_{i\in \mathcal{P}}  h_{\boldsymbol{\tau}_i}(\mathbf{Z}_i),
\end{equation}
where  $h_{\boldsymbol  \tau_i}$ is  the  density  of the  multinomial
probability    distribution   $\mathcal{M}(1;\boldsymbol\tau_i)$,   and
$\boldsymbol{\tau}_i  =  (\tau_{i1},\dots,\tau_  {iQ})$  is  a  random
vector containing  the variational  parameters to optimize.   The complete
set        of       parameters        $\boldsymbol        \tau       =
\left\{\tau_{iq}\right\}_{i\in\mathcal{P},q\in\mathcal{Q}}$ is what we are seeking
to obtain via the variational  inference. In the  case in
hand the  variational approach intuitively operates  as follows: each
$\tau_{iq}$ must  be seen as  an approximation of the  probability that
vertex $i$ belongs to cluster $q$, conditional on the data, that is,
$\tau_{iq}$ estimates $\mathbb{P}(Z_{iq} = 1 | \mathbf{K})$, under the
constraint   $\sum_{q}  \tau_{iq}=1$.    In  the   ideal   case  where
$\mathbb{P}(\mathbf{Z}|\mathbf{K})$     can    be     factorized    as
$\prod_i\mathbb{P}(\mathbf{Z}_i|\mathbf{K})$    and   the   parameters
$\tau_{iq}$  are  chosen  as  $\tau_{iq}  = \mathbb{P}(Z_{iq}  =  1  |
\mathbf{K})$, the K\"ullback-Leibler
divergence is null and the bound $\mathcal{J}$ reaches the log-likelihood.\\

The  following   proposition  gives  the  form  of   the  lower  bound
$\mathcal{J}$ to be maximized in order to estimate $\boldsymbol \tau$.
\begin{proposition} Let us assume that $R_{\boldsymbol\tau}$  can be factorized as
 in           \eqref{eq:R_factorized},           and let us denote
 $\mathcal{J}_{\boldsymbol\tau}\left(\mathbf{X},\mathbf{K}\right)   :=
 \mathcal{J}\left(\mathbf{X},\mathbf{K},R_{\boldsymbol\tau}(\mathbf{Z})\right)$.
 Then   $\mathcal{J}_{\boldsymbol  \tau}$  satisfies   the  following
 expression
 \begin{multline}
   \label{eq:J_factorized}
   \mathcal{J}_{\boldsymbol\tau}\left(\mathbf{X},\mathbf{K}\right) =
   c  -\sum_{\substack{i\in   \mathcal{P}  \\
       q\in\mathcal{Q}}} \tau_{iq}
   \log\tau_{iq}     +     \sum_{\substack{i\in    \mathcal{P}     \\
       q\in\mathcal{Q}}} \tau_{iq} \log \alpha_q \\
   - \left\|\boldsymbol \rho_{\boldsymbol\tau}(\mathbf{K})\right\|_{\ell_1}
   - \sum_{\substack{i,j\in\mathcal{P} , i\neq j \\     
  q,\ell\in\mathcal{Q}}}   \tau_{iq}     \tau_{j\ell}     \log
   2\lambda_{q\ell},
 \end{multline}
 where     $c$    does    not     depend    on
 $\boldsymbol\tau$                  and                  $\boldsymbol
 \rho_{\boldsymbol\tau}(\mathbf{K})=        (\rho_{\boldsymbol{\tau}_i
   \boldsymbol{\tau}_j}(K_{ij}))_{i,j\in  \mathcal{P}^2}$  is defined
 similarly as \eqref{eq:rho_func}, replacing $Z_{iq}$ by $\tau_{iq}$.
\end{proposition}

\begin{proof}Starting from  \eqref{eq:definition_J}, classical results
  on variational methods show that
  \begin{equation*}
    \mathcal{J}_{\boldsymbol\tau}\left(\mathbf{X},\mathbf{K}\right)
    =  \widehat{Q}_{\boldsymbol\tau}(\mathbf{K}) + \mathcal{H}(R_{\boldsymbol\tau}(\mathbf{Z})),
  \end{equation*}
  where  $\mathcal{H}(R_{\boldsymbol\tau}(\cdot))$ is  the  entropy of
  the        distribution       $R_{\boldsymbol\tau}(\cdot)$       and
  $\widehat{Q}_{\boldsymbol\tau}(\mathbf{K})$  is  the approximation
  of  the  complete log-likelihood conditional  expectation,  computed under  the  distribution  $R_{\boldsymbol\tau}$. Namely,
  \begin{equation}                               \label{eq:def_Q_chapo}
    \widehat{Q}_{\boldsymbol\tau}(\mathbf{K})                         =
    \mathbb{E}_{R_{\boldsymbol\tau}}
    \left\{\log\mathcal{L}_c(\mathbf{X},\mathbf{K},\mathbf{Z})
    \right\}       =      \sum_{\mathbf{Z}       \in      \mathcal{Z}}
    R_{\boldsymbol{\tau}}(\mathbf{Z})
    \log\mathcal{L}_c(\mathbf{X},\mathbf{K},\mathbf{Z}) .
  \end{equation}
  In     the    special     case     of    factorized     distribution
  \eqref{eq:R_factorized}, the entropy is
  \begin{equation*}
    \mathcal{H}(R_{\boldsymbol\tau}(\mathbf{Z})) = \sum_{i\in\mathcal{P}}
    \mathcal{H}(h_{\boldsymbol\tau_i}(\mathbf{Z}_i))                      =
    -\sum_{i\in\mathcal{P},q\in\mathcal{Q}} \tau_{iq} \log\tau_{iq}.
  \end{equation*}
  Moreover, 
  \begin{equation*}
    \widehat{Q}_{\boldsymbol\tau}(\mathbf{K})=\log\mathbb{P}(\mathbf{X}|
    \mathbf{K})+\mathbb{E}_{R_{\boldsymbol\tau}}[\log\mathbb{P}(\mathbf{K}
    |\mathbf{Z})]+\mathbb{E}_{R_{\boldsymbol\tau}}[\log\mathbb{P}(\mathbf{Z})].
 \end{equation*}
 Equation    \eqref{eq:J_factorized}     follows    via    Proposition
 \ref{prop:complete_likelihood},         by         using         that
 $\mathbb{E}_{R_{\boldsymbol\tau}}(Z_{iq})=\tau_{iq}$               and
 $\mathbb{E}_{R_{\boldsymbol\tau}}(Z_{iq}Z_{j\ell})=\tau_{iq}\tau_{j\ell}$.
\end{proof}

The  optimal approximate  distribution  $R_{\boldsymbol\tau}$ is  then
derived  by  direct  maximization of  $\mathcal{J}_{\boldsymbol\tau}$.
The      following       proposition      gives      the      estimate
$\widehat{\boldsymbol\tau}$ that solves the problem.

\begin{proposition} \label{prop:pointfixe}
Let $\boldsymbol\alpha$  and $\boldsymbol\lambda$
 be  known.  The  following fixed-point relationship  holds  for the
 optimal  variational  parameters  $\widehat{\boldsymbol\tau} =  \arg
 \max_{\boldsymbol{\tau}} \mathcal{J}_{\boldsymbol\tau}$
 \begin{equation}
   \label{eq:tauiq_laplace}
   \widehat{\tau}_{iq}     \varpropto  \alpha_q    
   \prod_{\substack{j\in\mathcal{P}\backslash         \{i\}        \\
       \ell\in\mathcal{Q}}}       \left(\frac{1}{2\lambda_{q\ell}}\exp
     \left\{-\frac{|K_{ij}|}{\lambda_{q\ell}}\right\}\right)^{\widehat{\tau}_{j\ell}} ,
 \end{equation}
where $\varpropto$ means that there is a scaling factor such that for any $i\in \mathcal{P}$, we have $\sum_q \widehat \tau_{iq}=1$.
\end{proposition}
\begin{proof}This  is just  an  adaptation to  the  Laplace case  of
 \citet[][Proposition 3]{2007_report_Mariadassou}.
\end{proof}

The   initial   value  of   $\boldsymbol\tau$   is   chosen  using   a
classification  algorithm such as  spectral clustering  \citep[see for
instance][]{2002_proc_Ng}.  As  a consequence, the  initial values for
$\tau_{iq}$  lie in  $\{0,1\}$. We then use  an  iterative procedure
setting   $\widehat{\boldsymbol   \tau}^{(m+1)}=g(\widehat{\boldsymbol
  \tau}^{(m)})$, where $g$ is  the function (implicitly defined above)
for which $\widehat{\boldsymbol \tau}$ is  a fixed point. Note that we
cannot ensure uniqueness  of the fixed point for  $g$, nor convergence
of this iterative procedure.  In practice, we can always use a maximal
number of iterations, and if  convergence has not occurred, we keep the
initial value  of $\boldsymbol \tau$  given by the  clustering method.
In  appendix~\ref{app:point_fixe} we  explain  that at  least in  the
affiliation   model~\eqref{eq:affiliation},  if  the   current  values
$K_{ij}^{(m)}$ of  the precision matrix  are small enough, and  if the
penalty        parameters        $\lambda_{\text{in}}^{-1}$        and
$\lambda_{\text{out}}^{-1}$  are well-chosen,  then uniqueness  of the
fixed  point is  ensured.  However,  such a  result does  not  hold
in the general case, which is one of the drawbacks of of the  variational
approach in this context.

\paragraph*{Estimation of   $\boldsymbol\alpha$ and $ \boldsymbol\lambda$.} 
The parameters $\boldsymbol\alpha$ and $ \boldsymbol\lambda$ have been
previously  considered as  known to  keep  the statement  as clear  as
possible.

Two different strategies may be used with respect to these parameters.
The  first approach  is  to fix  their  values.  Fixing  the value  of
$\boldsymbol\alpha$  comes  down   to  choosing  \emph{a  priori}  the
proportions of the groups, which is quite a common strategy in mixture
models. As for the choice of $ \boldsymbol\lambda$, this is equivalent
to    choosing    the    penalty    parameter   in    the    classical
\textsc{lasso}. Concerning general parameters $ \boldsymbol\lambda$, a
number  of   values  need   to  be  determined,   which  might   be  a
problem.      However      in      the     particular      affiliation
model~\eqref{eq:affiliation}, only $2$ parameters  have to be fixed: a
parameter $\lambda_{\mathrm{in}}$ that corresponds to a light penalty,
since  many intra-cluster  edges are  expected, and  another parameter
$\lambda_{\mathrm{out}}$ that fits with a heavier penalty, since we do
not expect  many inter-cluster edges.   This is typically the  kind of
strategy   that  will   be  used   for  numerical   applications  (see
Section~\ref{sec:numexp}).  More generally,  the matrix penalty can be
tuned to obtain a desired  quantity of inferred edges, or to constrain
the topology of the graph, e.g.  graphs with hubs.

The second strategy is to make use of the current inferred graph
to  estimate the  parameters.  The  basic  idea is  to include  this
estimation in the variational method.  Unfortunately, the maximization
of     $\mathcal{J}_{\boldsymbol    \tau}$    given     in    equation
\eqref{eq:J_factorized}    with    respect    to    $\boldsymbol\tau$,
$\boldsymbol\lambda$ and  $\boldsymbol\alpha$ at the same  time is not
possible.  To tackle  this problem, we use an  alternate strategy. The
parameter  $\boldsymbol{\tau}$  is   computed  with  the  fixed-point
relationship    \eqref{eq:tauiq_laplace}   for    fixed    values   of
$\boldsymbol\lambda$   and  $\boldsymbol\alpha$.   Then   we  maximize
$\mathcal{J}_{\boldsymbol \tau}$  with respect to $\boldsymbol\lambda$
and $\boldsymbol\alpha$, once  $R_{\boldsymbol\tau}$ is fixed (that is,
once $\boldsymbol{\tau}$  is fixed), as in  the following proposition.
We successively iterate these two steps until stabilization.

\begin{proposition}  For   fixed  values  of   $\boldsymbol\tau$,  the
  parameters         $\hat        {\boldsymbol{\alpha}}$,        $\hat
  {\boldsymbol{\lambda}}$                                    maximizing
  $\mathcal{J}_{\boldsymbol{\tau}}$  are  given
  by
  \begin{equation*}
    \forall  q,\ell \in  \mathcal{Q}, \  \hat{\alpha}_q  = \frac{1}{p}
    \sum_{i\in\mathcal{P}} \tau_{iq} \ \textrm{and} \
    \hat{\lambda}_{q\ell} = \frac{\sum_{i\neq j} \tau_{iq}\tau_{j\ell}
      |K_{ij}|}{\sum_{i\neq j}\tau_{iq}\tau_{j\ell}}.
  \end{equation*}
\end{proposition}

\begin{proof}
 Once terms that do not depend on the parameters of interest have been
 removed  from $\mathcal{J}_{\boldsymbol \tau}$, the problem becomes
 \begin{equation*}
   \hat \alpha_q = \argmax_{\alpha_q} \sum_{i}\tau_{iq}\log \alpha_q
   \ \textrm{and} \ \hat{\lambda}_{q\ell} = \argmax_{\lambda_{q\ell}}
   -           \sum_{i\neq          j}          \tau_{iq}\tau_{j\ell}
   \left(\frac{|K_{ij}|}{\lambda_{q\ell}}+\log 2\lambda_{q\ell}\right).
 \end{equation*}
 Null-differentiation   with  respect   to   $\alpha_q$  (under   the
 constraint   $\sum_q   \alpha_q=1$)   and  $\lambda_{q\ell}$   leads
 straightforwardly to the result.
\end{proof}

\subsection[A Lasso-like  method to estimate  the concentration matrix]{A Lasso-like  method to estimate  the concentration matrix  (the M-step)}

Now that we are able  to  compute the  approximate conditional  expectation
$\widehat{Q}_{\boldsymbol\tau}(\mathbf{K})$         defined         by
\eqref{eq:def_Q_chapo},  we  wish to  infer  the concentration  matrix
$\mathbf{K}$, assuming $\boldsymbol\tau$ is known.  This is the aim of
the \textsc{m}-step of our \textsc{em}--like strategy, that deals with
the    maximization    problem    $\arg    \max_{\mathbf{K}\succ    0}
\widehat{Q}_{\boldsymbol\tau}(\mathbf{K})$.

Using Proposition
\ref{prop:complete_likelihood} and the equality $\mathbb{E}_{R_{\boldsymbol\tau}}(Z_{iq}Z_{j\ell})=\tau_{iq}\tau_{j\ell}$,  it is a simple matter to rewrite the problem  as follows
\begin{equation}
 \label{eq:M_step}
 \widehat{\mathbf{K}}        =       \argmax_{\mathbf{K}\succ
   0} \left\{\frac{n}{2}\left(\log \det (\mathbf{K}) -
     \mathrm{Tr}(\mathbf{S}\mathbf{K})\right)                 -\left\|\boldsymbol
     \rho_{\boldsymbol \tau}(\mathbf{K}) \right\|_{\ell_1}\right\}. 
\end{equation}
Hence,  our  \textsc{m}--step  can  be  seen as  a  penalized  maximum
likelihood estimation problem, 
exactly like in 
%which  is the interpretation adopted in
\cite{2007_BS_Friedman,2008_JMLR_Banerjee}.  The likelihood considered
here   is  $\mathbb{P}(\mathbf{X}|\mathbf{K})$,   that  is,   the  likelihood
which corresponds to  the $n$ realizations of  the Gaussian vector  $X$ for a
given  concentration  matrix   $\mathbf{K}$.  The difference of our approach
lies in the complexity  of  the penalty  term, and  in slight
discrepancies as regards some constant factors.

\begin{remark}\label{rem:diagonal_solutions}
 Since we are using  a penalty term $1/\lambda_0$ on matrix $\mathbf
 K$'s diagonal elements, the solution to \eqref{eq:M_step} satisfies
\begin{equation}
  \label{eq:diagonal}
  \forall i \in \mathcal{P}, \quad \widehat {K}^{-1}_{ii} = S_{ii} +2/(n \lambda_0),
\end{equation}
when $\lambda_0^{-1} < n  |S_{ii}|/2$ for any $i \in \mathcal{P}$. 
Indeed,         the         sub-gradient        equation         is
$n/2(K^{-1}_{ii}-S_{ii})+\mathrm{sgn}(K_{ii})/\lambda_0=0$, and $K_{ii}\geq 0$ since it is the inverse of a conditional variance.
\end{remark}

Let  us now  look at  the  solution of  the \textsc{m}-step:  the
following   proposition    gives   an   equivalent    formulation   of
\eqref{eq:M_step} that is more likely  to be solved.  The result draws
its  inspiration  from \cite{2008_JMLR_Banerjee}.  

\begin{proposition}The maximization problem \eqref{eq:M_step} over the
  concentration  matrix $\mathbf{K}$ is  equivalent to  the following,
  dealing with the covariance matrix $\boldsymbol\Sigma$
  \begin{equation}
    \label{eq:M_step_optim}
    \widehat{\boldsymbol\Sigma} = \argmax_{ 
      \|     (\boldsymbol\Sigma      -     \mathbf{S})     \cdot\slash
      \mathbf{P}_{\boldsymbol\tau}\|_{\infty}\le
      1}\log\det(\boldsymbol\Sigma),
  \end{equation}
  where $\cdot\big\slash$ is the term-by-term division and
  \begin{equation*}
    \mathbf{P}_{\boldsymbol \tau} = (P_{\boldsymbol \tau_i \boldsymbol \tau_j})_{i,j\in\mathcal{P}}  \quad \text{ with }\quad
 P_{\boldsymbol \tau_i \boldsymbol \tau_j} =  \left\{ 
  \begin{array}{cc}
    2n^{-1}\sum_{q,\ell}\tau_{iq}\tau_{j\ell} \lambda_{q\ell}^{-1} & i\neq j,\\
2(n\lambda_0)^{-1} & i=j.
   \end{array} 
\right.
  \end{equation*}
\end{proposition}

\begin{remark}\label{rem:pen_diag}
 By penalizing the diagonal terms of the concentration matrix
 $\mathbf K$ in the initial problem, the set of matrices $\boldsymbol
 \Sigma $ over which we maximize our criterion contains, for instance,
 the matrix $\mathbf{S}+ 2/(n\lambda_0) I$, (where $I$ stands for the
 identity matrix).  Thus, provided that the value of penalty parameter
 $1/\lambda_0$ is set sufficiently high, this set contains positive
 definite matrices. This ensures that our estimator is always
 invertible. Obviously, when $\mathbf{S}$ is invertible, which is
 usually true for $n$ greater or equal than $p$, penalizing the
 diagonal terms 
 becomes futile. In this case $1/\lambda_0$ is set to zero.
\end{remark}

\begin{proof}The penalty  term in \eqref{eq:M_step} can  be written as
 follows
 \begin{equation*}
   \left\|\boldsymbol       \rho_{\boldsymbol       \tau}(\mathbf{K})
   \right\|_{\ell_1}           =          \sum_{q,\ell\in\mathcal{Q}}
   \sum_{ \substack{i,j\in\mathcal{P}  \\ i\neq j}}
   \frac{\left|K_{ij}     \right|}{\lambda_{q\ell}}  \tau_{iq}\tau_{j\ell}   
   + \sum_{i\in \mathcal{P}} \frac{|K_{ii}|}{\lambda_0}    =
   \sum_{q,\ell\in\mathcal{Q}}\left\|\mathbf{T}_{q\ell}\star\mathbf{K}\right\|_{\ell_1},
 \end{equation*}
 where    $\star$   is   the    term-by-term   product.     The   set
 $\left\{\mathbf{T}_{q\ell} \right\}_{q,\ell\in\mathcal{Q}}$ contains
 $p\times p$ symmetric matrices, defined, for each couple $(q,\ell)$,
 by
 \begin{equation*}
   \mathbf{T}_{q\ell} =  \left(T_{q\ell;ij}\right)_{i,j\in\mathcal{P}}  
   \quad \text{with} \quad \forall i\neq j, \quad T_{q\ell;ij}= \frac{\tau_{iq}\tau_{j\ell}}{\lambda_{q\ell}} \quad \text{ and } \quad T_{q\ell;ii} = \frac{1}{\lambda_{0}Q^2}. 
 \end{equation*}

 Let   us   now    use   the   fact   that   $\|\mathbf{A}\|_{\ell_1}
 =\max_{\|\mathbf{U}\|_\infty \leq 1} \mathrm{Tr}(\mathbf{A U})$, for
  a   given   matrix $\mathbf{A}$.   The   optimization   problem
 \eqref{eq:M_step} can now be written as
 \begin{equation*}
   \max_{\mathbf{K}\succ      0}      \min_{\left\{\mathbf{U}_{q\ell}      :
     \|\mathbf{U}_{q\ell}\|_\infty\leq 1\right\}} 
   \left\{  \frac{n}{2}  \log\det \mathbf{K} 
     - \mathrm{Tr}\left( 
       \frac{n}{2} \mathbf{S} \mathbf{K} + \sum_{q,\ell\in\mathcal{Q}} 
       \left(\mathbf{T}_{q\ell}\star\mathbf{K}\right)\mathbf{U}_{q\ell}
     \right)
   \right\},
 \end{equation*}
 since the trace  operator is linear. The dual  version of the above expression
 is  obtained  by  swapping  $\max$  and  $\min$.   The
 maximization   is  solved   by  differentiating   with   respect  to
 $\mathbf{K}$.  To do this, we recall that in our
 specific  case the  matrices  $\mathbf{T}$ are  symmetrical, and  thus
 $\mathrm{Tr}\left((\mathbf{T}\star\mathbf{K})\mathbf{U}\right)      =
 \mathrm{Tr}\left(\mathbf{K}(\mathbf{T}\star\mathbf{U})\right)$.
 Then,  applying the  usual  rules  for  the derivative of  the  trace  operator,
 null-differentiation with respect to $\mathbf{K}$ yields
 \begin{equation}
   \label{eq:Sigma_opt}
   \boldsymbol   \Sigma   :=   \mathbf   {K}^{-1}  =   \mathbf{S}   +
   \frac{2}{n}\sum_{q,\ell\in\mathcal{Q}}
   \left(\mathbf{U}_{q\ell}\star\mathbf{T}_{q\ell}\right).
 \end{equation}
 The dual problem therefore becomes
 \begin{equation*}
   \min_{\left\{\mathbf{U}_{q\ell}      :
       \|\mathbf{U}_{q\ell}\|_\infty\leq 1\right\}} 
   \left\{ 
     - \frac{n}{2}\log \det(\boldsymbol\Sigma) - \frac{np}{2} 
   \right\},
 \end{equation*}
 or in other words,
 \begin{equation*}
   \max_{\left\{\mathbf{U}_{q\ell}      :
       \|\mathbf{U}_{q\ell}\|_\infty\leq 1\right\}} 
   \log \det(\boldsymbol\Sigma).
 \end{equation*}
 Finally, we  need  to write  the  constraint as  a function  of
 $\boldsymbol\Sigma$  rather than  the  set $\{\mathbf{U}_{q\ell}\}$.
 In fact, we simply need to show that
 \begin{equation*}
   \left\{\mathbf{U}_{q\ell};    \forall    q,\ell    \in    \mathcal{Q},
     \|\mathbf{U}_{q\ell}\|_\infty\leq           1\right\}          =
   \left\{    \boldsymbol    \Sigma;   \left\|(\boldsymbol\Sigma    -
       \mathbf{S})     \cdot    \big\slash    \mathbf{P}_{\boldsymbol
         \tau}\right\|_{\infty} \leq 1 \right\},
 \end{equation*}
 which is straightforward  (see Appendix \ref{ap:proof_norm} for details).
\end{proof}

To  solve   \eqref{eq:M_step_optim}  and  thus   obtain  the  estimate
$\widehat{  \boldsymbol \Sigma}$, we  successively use  two coordinate
descent methods.   The first corresponds to  a block-wise strategy
suggested by \citeauthor{2008_JMLR_Banerjee}.   The second one is used
to  solve the resulting  \textsc{Lasso} problem  and was  suggested by
\cite{2007_AAS_Friedman}.
\\

Let us  first explain  the block-wise strategy.  For this  purpose, we
introduce the following  notation for $\widehat{ \boldsymbol \Sigma}$,
$\mathbf{S}$ and the penalty matrix $\mathbf{P}_{\boldsymbol \tau}$
\begin{equation}
  \label{eq:Sigma_block}
  \widehat{ \boldsymbol \Sigma} = \begin{bmatrix}
    \widehat{ \boldsymbol \Sigma}_{11} &  \widehat{\boldsymbol \sigma}_{12} \\
    \widehat{\boldsymbol \sigma}_{12}^\intercal &  \widehat{ \Sigma}_{22} \\
  \end{bmatrix}, \quad 
  \mathbf{S} = \begin{bmatrix}
    \mathbf{S}_{11} & \mathbf{s}_{12} \\
    \mathbf{s}_{12}^\intercal & S_{22} \\
  \end{bmatrix},\quad 
  \mathbf{P}_{\boldsymbol \tau} = \begin{bmatrix}
    \mathbf{P}_{11} & \mathbf{p}_{12} \\
    \mathbf{p}_{12}^\intercal & P_{22} \\
  \end{bmatrix},
\end{equation}
where  $ \widehat{  \boldsymbol  \Sigma}_{11}$, $\mathbf{S}_{11}$  and
$\mathbf{P}_{11}$      are      $(p-1)\times     (p-1)$      matrices,
$\widehat{\boldsymbol     \sigma}_{12}$,     $\mathbf{s}_{12}$     and
$\mathbf{p}_{12}$  are  $(p-1)$ length  column  vectors and  $\widehat
\Sigma_{22}$, $S_{22}$ and $P_{22}$ are real numbers. We have already remarked
(Remark~\ref{rem:diagonal_solutions})       that   the   solution      to
\eqref{eq:M_step_optim}      satisfies      $\widehat      \Sigma_{22}
=S_{22}+2/(n\lambda_0)$.  Moreover, using Sch\"ur  complement, the
vector $\widehat{\boldsymbol \sigma}_{12}$ satisfies
\begin{equation}
  \label{eq:M_step_optim_block}
  \widehat{\boldsymbol \sigma}_{12}       =       \argmin_{\left\{\mathbf{y}:      \|       (\mathbf{y}-\mathbf{s}_{12})\cdot\slash
    \mathbf{p}_{12}\|_\infty \leq 1\right\}}\left\{\mathbf{y}^\intercal    \widehat{\boldsymbol \Sigma}^{-1}_{11}\mathbf{y}\right\}.
\end{equation}
We have $\det(\widehat{\boldsymbol \Sigma})=
\det(\widehat{\boldsymbol \Sigma}_{11})(\widehat{
  \Sigma}_{22}-\widehat{\boldsymbol \sigma}_{12}^\intercal
\widehat{\boldsymbol \Sigma}_{11}^{-1}\widehat{\boldsymbol
 \sigma}_{12})$.  The full matrix $\widehat{\boldsymbol \Sigma}$ is
approximated in the following way: first, if required when $p$ is
greater than $n$, we initialize the procedure with $
S+2/(n\lambda_0)I$, where $\lambda_0>0$ is chosen so as to make $
S+2/(n\lambda_0)I$ invertible; secondly, we permute the columns
(and thus the rows) of $\widehat{\boldsymbol \Sigma}$ and iteratively
solve problems like \eqref{eq:M_step_optim_block} until convergence of
the procedure. This convergence is ensured by the following lemma.

\begin{lemma}\label{lem:cv_coord_descent_1}
 The  procedure which  starts  with a  positive  definite matrix  and
 iteratively updates the columns and rows of this matrix according to
 the  solutions of  \eqref{eq:M_step_optim_block}  converges to  the
 solution $\widehat{\boldsymbol \Sigma}$ of \eqref{eq:M_step_optim}.
\end{lemma}

\begin{proof}
 The  proof  relies  on \citet[][Theorem  3]{2008_JMLR_Banerjee}  and
 \citet[][Theorem      4.1]{2001_JOTA_Tseng}.      Convergence     of
 block-coordinate  descent  methods is  a  well-documented topic  in
 convex optimization  literature. Here, we have to bear in mind that using $\ell_1$-norm  penalty leads  to non-differentiable
 functions.    Thus,  we   rely   on  a   result  by   \citet[Theorem
 4.1]{2001_JOTA_Tseng}, which in our  case ensures the convergence of
 the procedure,  provided there  is at most  one solution  to each
 minimization  problem \eqref{eq:M_step_optim_block}.  This  point is
 proved in \citet[][Theorem 3]{2008_JMLR_Banerjee}.
\end{proof}

Then, starting  from a  result given in  \cite{2008_JMLR_Banerjee}, an
interpretation     of      \eqref{eq:M_step_optim_block}     as     an
$\ell_1$--penalized  problem  is  given  in  \citet{2007_BS_Friedman}.
This  $\ell_1$--penalized problem is  reminiscent of  the {\protect\sc
  Lasso} and  may thus be  solved using a coordinate  descent strategy
\citep{2007_AAS_Friedman}.   The  following  proposition enunciates  a
result   similar    to   those    obtained   in   \citet[equation
(6)]{2008_JMLR_Banerjee} and \citet[equation (2.4)]{2007_BS_Friedman},
although with a  more general penalty term and  a factor $\frac{1}{2}$
that differs.  Since  none of these articles gives  an explicit proof
for this result, it is fitting that we provide our own proof here.

\begin{proposition}Solving \eqref{eq:M_step_optim_block} is equivalent
 to solving the dual problem
 \begin{equation}
   \label{eq:M_step_lasso}
   \widehat{\boldsymbol\beta} =   \argmin_{\boldsymbol \beta} \left\|\frac{1}{2}  \widehat{\boldsymbol \Sigma}_{11}^{1/2} \boldsymbol\beta -
     \widehat{\boldsymbol
       \Sigma}_{11}^{-1/2}\mathbf{s}_{12}\right\|_2^2    +    \left\|
     \mathbf{p}_{12} \star 
     \boldsymbol \beta \right\|_{\ell_1},
 \end{equation}
 where     solution     $\widehat{\boldsymbol    \sigma}_{12}$     to
 \eqref{eq:M_step_optim_block}  and $\widehat  {\boldsymbol\beta}$ to
 \eqref{eq:M_step_lasso} are linked through
 \begin{equation}
   \label{eq:w12_to_beta}
   \widehat{\boldsymbol \sigma}_{12} = \widehat{\boldsymbol \Sigma}_{11}\widehat {\boldsymbol\beta}/2.
 \end{equation}
\end{proposition}

\begin{proof}Problem  \eqref{eq:M_step_optim_block} can be  written as
 follows, by splitting the constraint:
 \begin{equation*}
   \left\{\begin{array}{rcl}
       & \min_{\mathbf{y}} \mathbf{y}^\intercal \widehat{\boldsymbol \Sigma}_{11}^{-1} \mathbf{y} & \\
       \text{subject   to  }  &   -(\mathbf{p}_{12})_i  \leq   y_i  -
       (\mathbf{s}_{12})_i - (\mathbf{p}_{12})_i \leq 0, & \forall i=1,\dots,p-1,\\
       \text{or } &  -(\mathbf{p}_{12})_i  \leq  - y_i  +
       (\mathbf{s}_{12})_i - (\mathbf{p}_{12})_i \leq 0, & \forall i=1,\dots,p-1.
     \end{array}\right.
 \end{equation*}
 Let  us introduce  $L$  the so-called  Lagrangian,  with vectors  of
 Lagrange             coefficients             denoted             by
 $\boldsymbol\beta^1=(\beta^1_i)_{i\le                           p-1},
 \boldsymbol\beta^2=(\beta^2_i)_{i\le   p-1}$   with   nonnegative
 entries.        Also, let        $\boldsymbol\beta        =
 \boldsymbol\beta^2-\boldsymbol\beta^1$. The  Lagrange version of the
 above problem is
 \begin{equation}
   \label{eq:lagrange_form}
   \min_{\mathbf{y}} \left\{\mathbf{y}^\intercal \widehat{\boldsymbol \Sigma}_{11}^{-1} \mathbf{y} + \max_{\boldsymbol\beta}L(\boldsymbol\beta) \right\},
 \end{equation}
 where, in the present case, $L$ is given by
 \begin{equation*}
   L(\boldsymbol\beta)  =  
   \sum_{i} \beta_i^1 \left(y_i - (\mathbf{s}_{12})_i - 
     (\mathbf{p}_{12})_i\right) + \sum_{i} \beta_i^2 \left( - y_i +
     (\mathbf{s}_{12})_i - (\mathbf{p}_{12})_i \right),
 \end{equation*}
 The    coefficients   $\beta_i^1$    and    $\beta_i^2$   maximizing
 $L(\boldsymbol\beta)$ are  null when the  constraints are satisfied,
 and   for  each   index  $i$,   at  least   one   coefficient  among
 $\{\beta_i^1,\beta_i^2\}$ is zero.  Then
 \begin{equation*}
   \| \boldsymbol \beta \|_{\ell_1} = \sum_{i} \left| \beta_i \right|
   = \sum_{i} \left(\beta_i^1 + \beta_i^2\right).
 \end{equation*}
 Meanwhile,  consider the  dual problem  of \eqref{eq:lagrange_form},
 swapping  $\min$ and $\max$:  the solution  that minimizes  the dual
 problem  with respect  to $\mathbf{y}$  satisfies  the null-gradient
 hypothesis.   We  obtain  $2 \widehat{\boldsymbol  \Sigma}^{-1}_{11}
 \mathbf{y} -\boldsymbol\beta = 0$, that is $\mathbf{y} = \frac{1}{2}
 \widehat{\boldsymbol  \Sigma}_{11}  \boldsymbol\beta$ (which  proves
 equation  \eqref{eq:w12_to_beta}).  Introducing  this result  in the
 dual of \eqref{eq:lagrange_form}, we get
 \begin{equation*}
   \max_{\boldsymbol\beta}-\frac{1}{4}      \boldsymbol\beta^\intercal
   \widehat{\boldsymbol         \Sigma}_{11}\boldsymbol\beta        +
   \mathbf{s}_{12}^\intercal      \boldsymbol\beta      -      \sum_i
   \left(\beta_i^1 + \beta_i^2 \right) (\mathbf{p}_{12})_i, 
 \end{equation*} 
 also equivalent to 
 \begin{equation*}
   \min_{\boldsymbol\beta}   \frac{1}{4}   \boldsymbol\beta^\intercal
   \widehat{\boldsymbol         \Sigma}_{11}\boldsymbol\beta        -
   \mathbf{s}_{12}^\intercal            \boldsymbol\beta            +
   \left\|\mathbf{p}_{12}\star \boldsymbol\beta \right\|_{\ell_1}.
 \end{equation*}
 Expressing this quantity by using the Euclidean norm achieves the
 proof.
\end{proof}

Hence, the column $\widehat{\boldsymbol \sigma}_{12}$ of the estimated
covariance  matrix  $\widehat{\boldsymbol   \Sigma}$  is  computed  by
solving the  {\protect \sc Lasso}  problem \eqref{eq:M_step_lasso}
using another coordinate descent method.

\begin{lemma}\label{lem:coord_descent_2}
  The solution to \eqref{eq:M_step_lasso}  is computed by updating the
  $j$th coordinate of $\widehat {\boldsymbol\beta}$ via
  \begin{equation}\label{eq:beta}
    \widehat \beta_j = 2 S\left((\mathbf{s}_{12})_j - \frac{1}{2} \sum_{k\neq j} (\widehat{\boldsymbol\Sigma}_{11})_{jk}
      \widehat \beta_k \; ; \; (\mathbf{p}_{12})_j \right) / (\widehat{\boldsymbol\Sigma}_{11})_{jj},
  \end{equation}
  where   $S(x   ;\rho)    =   \mathrm{sgn}(x)(|x|-\rho)_+$   is   the
  soft-thresholding operator.

  Moreover,  the procedure  which iteratively  updates the  entries of
  vector  $\widehat{\boldsymbol  \sigma}_{12}  =  \widehat{\boldsymbol
    \Sigma}_{11}\widehat   {\boldsymbol\beta}/2$   according  to   the
  solutions $\widehat {\boldsymbol\beta}$ of \eqref{eq:beta} converges
  to the solution of \eqref{eq:M_step_optim_block}.
\end{lemma}
  
\begin{proof}
  The proof of this lemma is postponed to Appendix \ref{ap:pathwise}.
\end{proof}

Finally, the  estimate of the matrix of  concentration $\mathbf{K}$ is
recovered  by inverting $\widehat{\boldsymbol  \Sigma}$, which  can be
done at  low computational  cost (see appendix  \ref{ap:inversion} for
details).    Hence,  we   solve  the   initial   maximization  problem
\eqref{eq:M_step} that defines the \textsc{m}-step of our algorithm.
\\

Implementation  of the  full  \textsc{em} algorithm  is  outlined  in
Algorithm~\ref{algo:1}.

\begin{algorithm}\label{algo:1}
\dontprintsemicolon
\While{$\widehat{Q}_{\boldsymbol\tau}(\widehat{\mathbf{K}}^{(m)})$ has
  not stabilized}{

 \BlankLine
 \BlankLine
 \CommentSty{//THE E-STEP: LATENT STRUCTURE INFERENCE}\;
 \eIf{$m = 1$}{
   \CommentSty{// First pass}\;
   Apply   spectral  clustering   on  the   empirical  covariance
   $\mathbf{S}$ to initialize $\widehat{\boldsymbol\tau}$\;
 }{
   Compute   $\widehat{\boldsymbol\tau}$   with   the   fixed-point
   relationship \eqref{eq:tauiq_laplace}, using $\widehat{\mathbf{K}}^{(m-1)}$\;
 }
 \BlankLine
 \BlankLine

 \CommentSty{//THE M-STEP: NETWORK INFERENCE}\; 
 Construct the penalty matrix $\mathbf{P}$ according to $\widehat{\boldsymbol\tau}$\;
 \While{$\widehat{\boldsymbol \Sigma}^{(m)}$ has not stabilized}{

   \For{each column of $\widehat{\boldsymbol \Sigma}^{(m)}$}{
     Compute  $\widehat{\boldsymbol \sigma}_{12}$ by  solving  the \textsc{lasso}--like  problem
     with path-wise coordinate optimization\;
   }
 }
 Compute $\widehat{\mathbf{K}}^{(m)}$ by block inversion of $\widehat{\boldsymbol \Sigma}^{(m)}$\;

 \BlankLine
 \BlankLine
 $m\leftarrow m+1$\;
}
\BlankLine
\caption{The full \textsc{em}--like algorithm}
\label{algo:main}
\end{algorithm}

\subsection{Choice of penalty parameters}\label{sec:choice}

As previously stated, the penalty parameters $\boldsymbol\lambda $ may
be estimated in the \textsc{e}-step of the algorithm (see
subsection~\ref{sec:Estep}). However, this choice is not necessarily
optimal for the estimation of $\mathbf{K}$, and other choices might
in practice lead to a better solution.  A good strategy is to keep the
estimated value of $\boldsymbol \lambda$ in the \textsc{e}-step that
leads to the estimation of $\boldsymbol\tau$, and to impose another
value of $\boldsymbol \lambda$ during the \textsc{m}-step. In this
part, we indicate a possible choice for the penalty parameters to use
in the \textsc{m}-step, ensuring a small error on the connectivity
components of the estimated graph.

Let  us   first  introduce  some   notation.   For  any   node  $i\in
\mathcal{P}$, let $C_i$ denote  the connectivity component of node $i$
in  the true  underlying conditional  dependency graph,  and $\widehat
C_i$ the corresponding component resulting from the estimate $\widehat
{\mathbf{K}}$ of  this graph  structure. The following  proposition is
based on \citet[][Theorem 2]{2006_AS_Meinshausen} and \citet[][Theorem
2]{2008_JMLR_Banerjee}.

\begin{proposition}
 Fix   some  $\varepsilon>0$  and   choose  the   penalty  parameters
 $\boldsymbol \lambda$ such that, for all $q,\ell \in \mathcal{Q}$, 
 \begin{equation}\label{eq:lambdachoice}
   2p^2 F_{n-2}\left(\frac 2 {n\lambda_{q\ell}} \left( \max_{i\neq j}
       S_{ii}S_{jj}   -  \frac   1  {\lambda_{q\ell}^2}\right)^{-1/2}
     (n-2)^{1/2}\right)\le \varepsilon,
 \end{equation}
 where $1-F_{n-2}$  is the  c.d.f.  of Student's $t$-distribution
 with $n-2$ degrees of freedom. Then
\begin{equation}\label{eq:error}
  \mathbb{P}(\exists k, \widehat C_k \nsubseteq C_k) \le \varepsilon.
\end{equation}
\end{proposition}

\begin{proof}
 Here we simply indicate  the  main differences  between  the proof  of
 \citet[][Theorem  2]{2008_JMLR_Banerjee} and  what is  valid  in our
 context.   Note that according  to \eqref{eq:M_step},  the estimator
 $\widehat{\mathbf{K}}$  must satisfy the  following sub-gradient
 equation
\begin{equation*}
  \forall i\neq j, \quad  \frac n 2 \left(\widehat K_{ij}^{-1}-S_{ij}\right) -\left(\sum_{q,\ell} \frac{Z_{iq}Z_{j\ell}}{\lambda_{q\ell}} \right)\nu_{ij}=0
\end{equation*}
where $\nu_{ij}\in \textrm{sgn}(\widehat K_{ij})$.
Following  the proof  of  \citet[][Theorem 2]{2008_JMLR_Banerjee},  we
easily get
\begin{equation*}
  \mathbb{P}(\exists k, \widehat C_k \nsubseteq  C_k) \le p^2  \max_{i\in \mathcal{P}, j\notin C_i}  \mathbb{P}\left( \frac n 2 |S_{ij}| \ge \sum_{q ,\ell} \frac{Z_{iq}Z_{j\ell} } {\lambda_{q\ell}}\right).
\end{equation*}
Performing some computations involving the correlation between variables $X_i$ and $X_j$, we also obtain
\begin{equation*}
  \mathbb{P}(\exists k, \widehat C_k \nsubseteq  C_k) \le 2 p^2   \max_{q,\ell \in \mathcal{Q}} F_{n-2}\left(\frac{2(n-2)^{1/2}} {n\lambda_{q\ell}} \left( \max_{i\in \mathcal{P}, j\notin C_i} S_{ii}S_{jj} - \frac 1 {\lambda_{q\ell}^2}\right)^{-1/2} \right),
\end{equation*}
which entails the conclusion.
\end{proof}
 
\begin{remark}
 Following  \cite{2008_JMLR_Banerjee}, note that  in order  to ensure
 \eqref{eq:lambdachoice},  it   is  enough  to   choose  the  penalty
 parameter  $\boldsymbol  \lambda$ such  that,  for  all $q,\ell  \in
 \mathcal{Q}$,
\begin{equation*}
  \lambda_{q\ell}(\varepsilon) \ge \frac{2}{n}  \left(n-2+t^2_{n-2}\left(\frac{\varepsilon}{2p^2}\right)\right)^{1/2} \left( \max_{i\neq j} S_{ii}S_{jj}\right)^{-1/2}  t_{n-2}\left( \frac{\varepsilon}{2p^2}\right)^{-1},
\end{equation*}
where   $t_{n-2}(u)$  is  the   $(1-u)$-quantile  of   Student's
$t$-distribution     with      $(n-2)$     degrees     of     freedom, i.e. $F_{n-2}(t_{n-2}(u))= u$.
\end{remark}

\begin{remark}
 Inequality  \eqref{eq:lambdachoice} does not take into account that
 different penalty parameters are used for different hidden
 classes $q,\ell  \in \mathcal{Q}$.   An adaptation of  the preceding
 strategy is to use  current values $\mathbf {Z}^{(m)}$ obtained from
 the probabilities $ \boldsymbol  {\tau}^{(m)}$ of the hidden classes
 and  to choose    the   current   penalty    parameters   $\boldsymbol
 \lambda^{(m)}$ accordingly.  More precisely, let us set, for instance
 \begin{equation*} \forall i\in \mathcal{P}, \quad 
   Z_{iq}^{(m)}= \left\{
     \begin{array}{cc}
       1 & \text{if } q =\argmax_{\ell}\tau_{i\ell}^{(m)}\\
       0 & \text{otherwise}.
     \end{array}
\right.
 \end{equation*}
Then, when
\begin{equation}\label{eq:currentlambdachoice}
  2p^2  F_{n-2}\left(\frac 2
    {n\lambda_{q\ell}^{(m)}}    \left(    \max_{\substack{i\neq   j\\
          Z_{iq}^{(m)}Z_{j\ell}^{(m)}=1}}  S_{ii}S_{jj}   -  \frac  1
      {(\lambda_{q\ell}^{(m)})^2}\right)^{-1/2} (n-2)^{1/2}\right)\le
  \varepsilon,
\end{equation}
for all  $q,\ell \in \mathcal{Q}$,  the  current estimate
$\widehat {\mathbf{K}}^{(m)}$  of the dependency  graph will approximately
satisfy   \eqref{eq:error}.    Moreover,   in   order   to   ensure
\eqref{eq:currentlambdachoice},  it  is  enough  to choose,  for  all
$q,\ell \in \mathcal{Q}$,
\begin{equation}
  \label{eq:autolambda}
  \lambda_{q\ell}^{(m)}(\varepsilon) \ge \frac{2}{n}  \left(n-2+t^2_{n-2}\left(\frac{\varepsilon}{2p^2}\right)\right)^{1/2} \left( \max_{\substack{i\neq j\\ Z_{iq}^{(m)}Z_{j\ell}^{(m)} =1}} S_{ii}S_{jj}\right)^{-1/2}  t_{n-2}\left( \frac{\varepsilon}{2p^2}\right)^{-1}.
\end{equation}
\end{remark}

Typically, the kind of values obtained with \eqref{eq:autolambda} will
lead  to  large penalties  and,  consequently,  to \emph{very}  sparse
graphs:  practically, more  informative  networks can  be obtained  by
replacing  the  term  $\varepsilon/2p^2$ in  \eqref{eq:autolambda}  by
greater values. In any  cases, \eqref{eq:autolambda} should be seen as
a starting value.

\section{Link       with       Meinshausen      and       B\"uhlmann's
 approach} \label{sec:pseudo-likelihood}

We should also like to  fill the  gap between, on the one hand   solving
\eqref{eq:M_step}      and, on the other,      the      approach     proposed      in
\cite{2006_AS_Meinshausen}, where $p$ independent penalized regression
problems are solved using the  {\protect \sc Lasso}.  In fact, we shall show
that \citeauthor{2006_AS_Meinshausen}'s  approach is equivalent to
maximizing  the penalized  \emph{pseudo} log-likelihood  corresponding to
the size-$n$ sample of the multivariate Gaussian vector $X$ on the set
of non symmetric matrices.  Let us denote as $\widetilde{\mathcal{L}}$
this pseudo-likelihood, defined by
\begin{equation*}
  \log \widetilde{\mathcal{L}}(\mathbf{X};\mathbf{K})= \sum_{i\in\mathcal{P}}
  \left(\sum_{k=1}^n   \log  \mathbb{P}(X_i^k|X_{\mathcal{P}\backslash
      i}^k;\mathbf{K}_i)\right),
\end{equation*}
where $X^k_{\mathcal{P}\backslash i}$ is  the $k$th realization of the
Gaussian  vector  $X$, once   the  $i$th  coordinate has been removed. In  this
section,  the  $\ell_1$-norm of  matrices  is  restricted to  off-diagonal
elements   only,    that   is,   $\|\mathbf   A\|_{\ell_1}=\sum_{i\neq
  j}|A_{ij}|$.

\begin{proposition}
 Consider  the solution $\widehat  {\mathbf{K}}^{\textrm{pseudo}}$ to
 the penalized pseudo-likelihood problem
 \begin{equation}\label{eq:pen_pseudolikelihood}
   \widehat  {\mathbf{K}}^{\textrm{pseudo}} =\argmax_{\{K_{ij}, i\neq
     j\}} \; \; \log \widetilde{\mathcal{L}}(\mathbf{X};\mathbf{K}) -
   \|\mathbf P\star \mathbf K\|_{\ell_1},
 \end{equation}
 (whose   diagonal    is   fixed)   and    the   solution   $\widehat
 {\mathbf{K}}^{\textrm{MB}}$  given in  \cite{2006_AS_Meinshausen} to
 the  $p$ different  regression  problems, using  the matrix  penalty
 $2\mathbf{P}/n$. The two solutions have exactly the same null
 entries.
\end{proposition}

\begin{proof}
 Denote    by    $\mathbf{K}_{\backslash    i\backslash    i}$    and
 $\mathbf{S}_{\backslash i\backslash  i}$, respectively, the matrices
 $\mathbf{K}$ and $\mathbf{S}$ once their $i$th row and $i$th
 column have been removed.      Moreover,      $\mathbf{K}_{i\backslash     i}$     and
 $\mathbf{S}_{i\backslash i}$ are the $i$th rows of the matrices with
 the $i$th term removed.  After some routine computations, and using
 classical   results   for    Gaussian   multivariate   vectors   (see
 Appendix~\ref{ap:pseudo-lik}), it can be shown that
\begin{equation}
 \label{eq:pseudolikelihood_result}
 \log \widetilde{\mathcal{L}}(\mathbf{X};\mathbf{K}) = 
 \frac{n}{2}\sum_{i\in\mathcal{P}}\left(      \log      K_{ii}     -
   K_{ii}S_{ii} -  2 \mathbf{S}_{i\backslash i}\mathbf{K}_{i\backslash i}
   -  \frac{1}{K_{ii}} \mathbf{K}_{i\backslash  i} \mathbf{S}_{\backslash
     i\backslash i} \mathbf{K}_{i\backslash i}^\intercal\right) + c,
\end{equation}
where $c$  does not  depend on $\mathbf{K}$.   Thus, if we  forget the
symmetry constraint on  $\mathbf{K}$, maximizing the pseudo-likelihood
\eqref{eq:pseudolikelihood_result}  with respect  to  the non-diagonal
entries of $\mathbf{K}$ is  equivalent to $p$ independent maximization
problems   with  respect   to  each   column  $\mathbf{K}_{i\backslash
 i}^\intercal$.    Consider,   for  instance,   the   last  column   of
$\mathbf{K}$,  that   is,  for  $i=p$,   and  the  relative   term  in
\eqref{eq:pseudolikelihood_result}.  This term can be written as
\begin{multline*}
  -\frac{n}{2K_{22}}   \left(   2   K_{22}   \mathbf{s}_{12}^\intercal
    \mathbf{K}_{i\backslash i}^\intercal + \mathbf{K}_{i\backslash
      i}^\intercal \mathbf{S}_{11} \mathbf{K}_{i\backslash i}\right) \\
  =     -      \frac{n}{2K_{22}}     \left\|     \mathbf{S}_{11}^{1/2}
    \mathbf{K}_{i\backslash        i}^\intercal        +        K_{22}
    \mathbf{S}_{11}^{-1/2}\mathbf{s}_{12}\right\|_2^2 + c',
\end{multline*}
where    we    use    the    block-wise   notation    defined    above
\eqref{eq:Sigma_block}.    The   term   $C'$   does  not   depend   on
$\mathbf{K}_{i\backslash  i}$,  which is  the  current  column of  the
concentration   matrix   to  infer.    Namely,   $c'   =  -   K_{22}^2
\mathbf{s}_{12}^\intercal \mathbf{S}_{11}^{-1}\mathbf{s}_{12}$.

Consider   now   the   penalized   version   of   the   log-likelihood
\eqref{eq:pen_pseudolikelihood}:  we  wish   to  solve  $p$  penalized
problems of minimization as defined  above, which can be written
as follows
\begin{equation}
  \label{eq:pseudolikelihood_penalised}
  \min_{\boldsymbol\beta}   \left\|  \mathbf{S}_{11}^{1/2}\boldsymbol\beta  +
    K_{22} \mathbf{S}_{11}^{-1/2}\mathbf{s}_{12}\right\|_2^2 + \frac{2K_{22}}{n} \left\|\mathbf{p}_{12} \star \boldsymbol\beta \right\|_{\ell_1}.
\end{equation}

\citeauthor{2006_AS_Meinshausen}  wish  to  solve  $p$  {\protect  \sc
  Lasso}-problems, for instance for the last variable $p$,
\begin{equation}
  \label{eq:meinshausen_problem}
  \min_{\boldsymbol\alpha}  \frac{1}{n}\left\| \mathbf{X}_p -\mathbf{X}_{\backslash
      p}\boldsymbol\alpha                \right\|_2^2                +
  \left\|2n^{-1}\mathbf{p}_{12}\star \boldsymbol\alpha \right\|_{\ell_1},
\end{equation}
where  $\mathbf{X}_p$   is  the  $p$th  column   of  $\mathbf{X}$  and
$\mathbf{X}_{\backslash  p}$ is the  matrix of  data the
$p$th  column  has been removed (note  that   we  adapted  the penalization term
corresponding to the framework developed here).

The  minimum is  reached in  \eqref{eq:pseudolikelihood_penalised} for
null-differentiation, and we get
\begin{equation*}
  2   \mathbf{S}_{11}   \boldsymbol\beta    +   2   K_{22}   \mathbf{s}_{12}^\intercal   +\frac{2K_{22}}{n}
  \mathbf{p}_{12}\star\nu =0,
\end{equation*}
where    $\nu\in\mathrm{sign}(\boldsymbol\beta)$.    The    same   for
\eqref{eq:meinshausen_problem}, and we get
\begin{equation*}
  \frac{2}{n}            \mathbf{X}_{\backslash           p}^\intercal
  \mathbf{X}_{\backslash p}\boldsymbol\alpha  - \frac{2}{n} \mathbf{X}_p^\intercal
  \mathbf{X}_{\backslash p}+ 2n^{-1}\mathbf{p}_{12}\star\gamma =0,
\end{equation*}
where  $\gamma\in\mathrm{sign}(\boldsymbol\alpha)$.   Now,  just  note
that $n^{-1}\mathbf{X}_{\backslash p}^\intercal \mathbf{X}_{\backslash
  p}   =   \mathbf{S}_{11}$   and   $n^{-1}\mathbf{X}_{p}   ^\intercal
\mathbf{X}_{\backslash  p} = \mathbf{s}_{12}^\intercal$,  and problems
\eqref{eq:pseudolikelihood_penalised}                               and
\eqref{eq:meinshausen_problem}    are   equivalent,    provided   that
$\boldsymbol\alpha=-\boldsymbol\beta/K_{22}$.

Thus, the columns of the concentration matrix (with a removed diagonal
term)  inferred from the  penalized maximum  pseudo-likelihood problem
\eqref{eq:pen_pseudolikelihood},  and   those  inferred   with 
\citeauthor{2006_AS_Meinshausen}'s  approach,  share  exactly the  same
null-entries, that is, the same network of conditional dependencies.
\end{proof}

%%% Local Variables:
%%% mode: latex
%%% TeX-master: "CJC.tex"
%%% End:

%% file: sec3.tex
\section{Numerical experiments}\label{sec:numexp}

In  this section we  present numerical  experiments on both synthetic
data, to investigate how well the proposed selection procedure behaves, 
and real data, to demonstrate  the practical use  of GGM covariance
selection with latent structure.
In the remainder of this section we focus on an affiliation model
\eqref{eq:affiliation}, the
choice of the penalty being made in line with Section~\ref{sec:choice}.
More precisely, we fix the ratio
$\lambda_{\text{in}}/\lambda_{\text{out}}=1.2$ and either let the
value $1/\lambda_{\text{in}}$ vary when considering precision/recall
curves for synthetic data, or fix this parameter according to \eqref{eq:autolambda} when
dealing with real data.

\subsection{Synthetic  data}\label{subsec:setting}

We  perform numerical  experiments to  assess the  performance  of our
approach (\texttt{SIMoNe}, Statistical Inference for
Modular Network) and compare it to already existing methods for GGM
covariance selection:  \texttt{GLasso}  \citep{2007_BS_Friedman}  and
\texttt{GeneNet} \citep{2005_SAGMB_Schafer}.

Data synthesis in our framework requires  the simulation  of a
structured sparse  inverse covariance  matrix.   
% The simulated  matrix should  be
% positive  definite.  
%  The  model   described  in  this  paper  can  be
% interpreted  in a  Bayesian framework  where the  concentration matrix
% elements  are  assumed   to  be  drawn  from  a   mixture  of  Laplace
% distribution.  This  assumption is efficient  for encouraging sparsity
% but  not  for  ensuring  positive definiteness  of  the  concentration
% matrix. We thus have to rely on an alternative simulation strategy.
%To simulate  a structured sparse  inverse covariance matrix, 
To this aim,  we first
simulate a graph with an  affiliation structure.  We consider a simple
binary affiliation model where two types of edges exist: edges between
nodes of the same class  and edges between nodes of different classes.
The binary incidence matrix of the graph is transformed by randomly
flipping the sign of some elements in order to simulate both
positively and negatively correlated variables.
Positive definiteness  of this  matrix is  ensured by  adding a
large enough constant to the diagonal. 
The matrix is then further normalized to have a diagonal
of ones.  A Gaussian sample of
size $n$  with zero  mean and the above covariance  matrix is then  simulated 50
times.   The results  we present  below are  averaged over  the 50
samples. At the end of this section we discuss the performances of our method when
there is no latent structure on the data.
\begin{figure}[htbp!]
  \centering
  \begin{tabular}{@{}c@{}c@{}c@{}}
    \includegraphics[width=.33\textwidth]{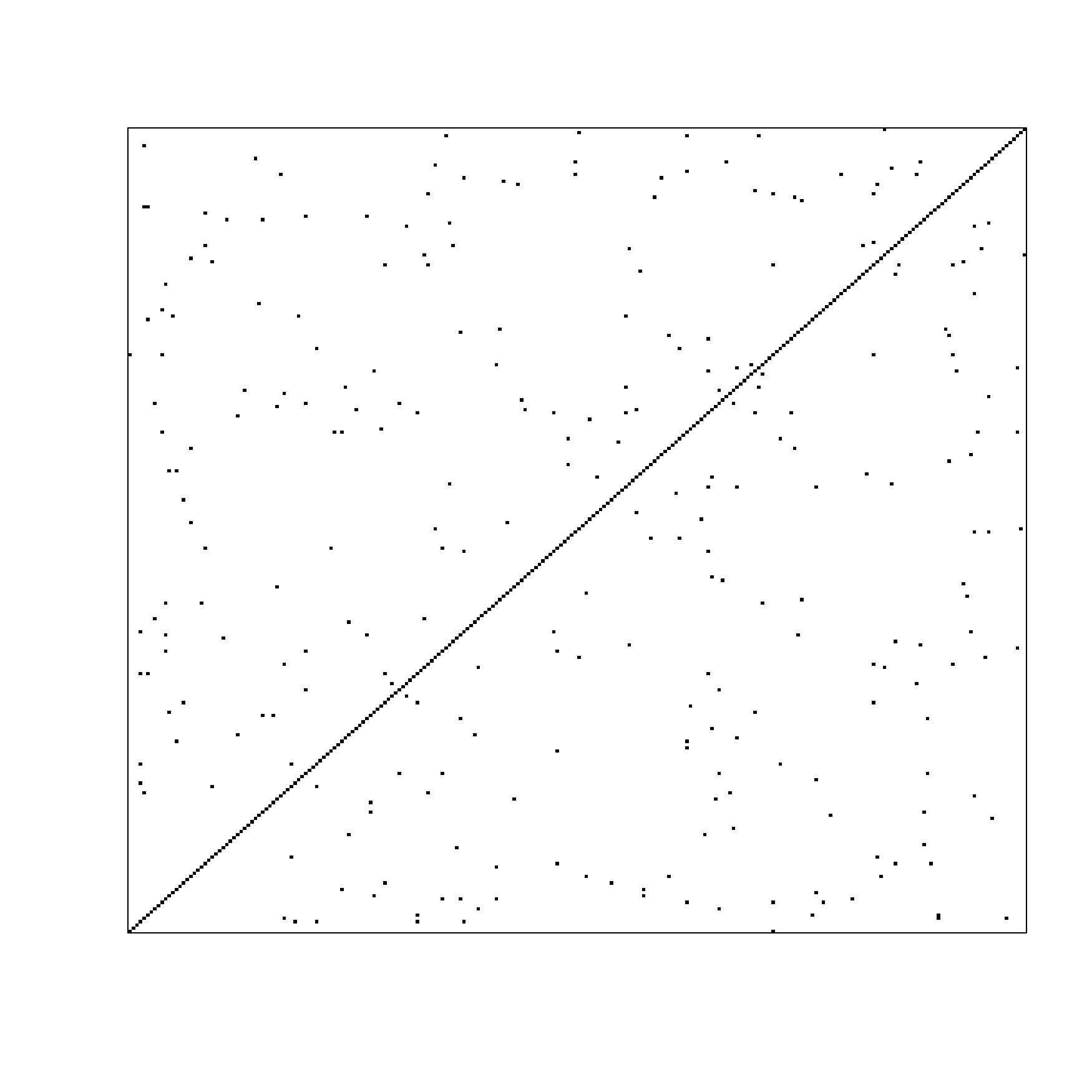}
    &
    \includegraphics[width=.33\textwidth]{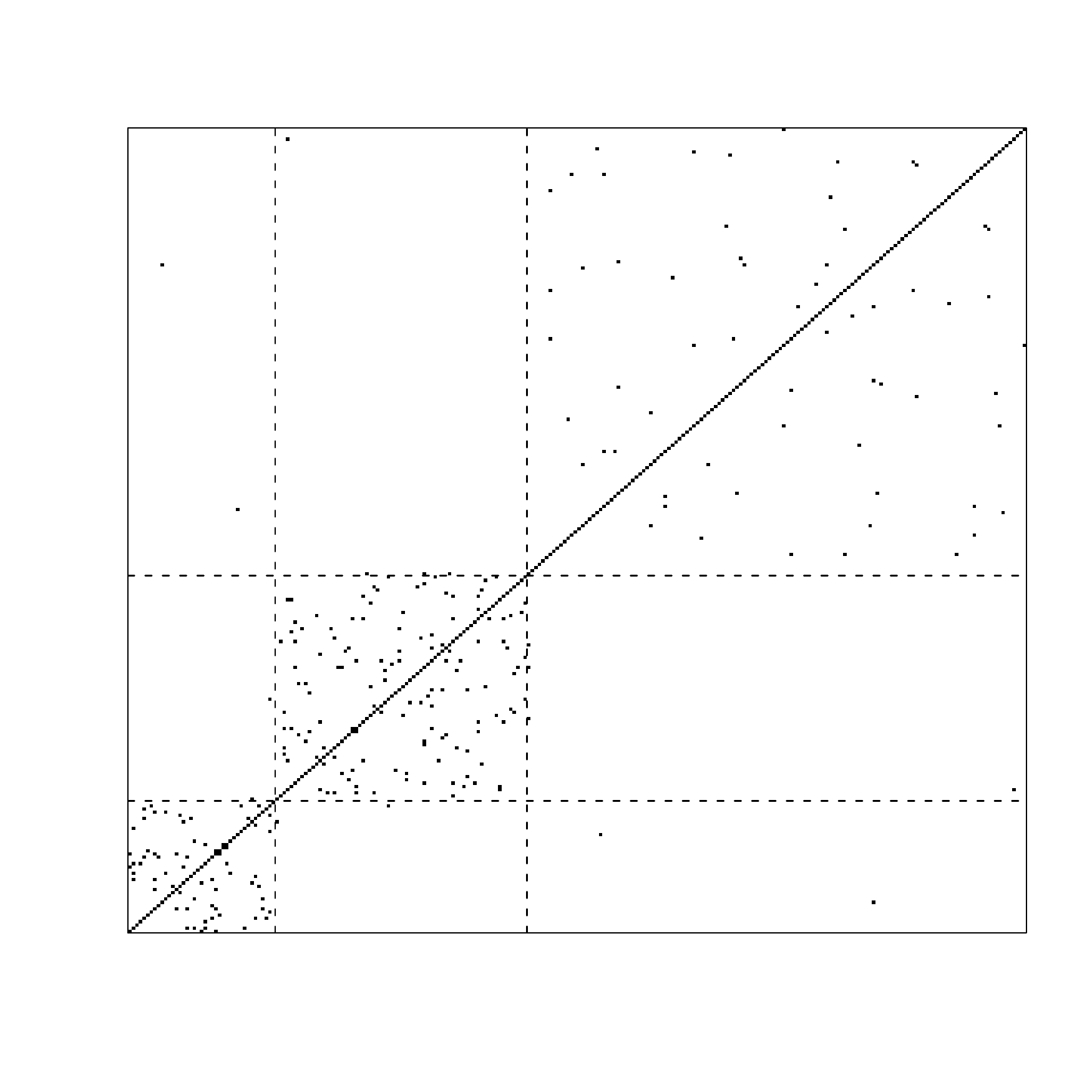}
    & 
    \includegraphics[width=.33\textwidth]{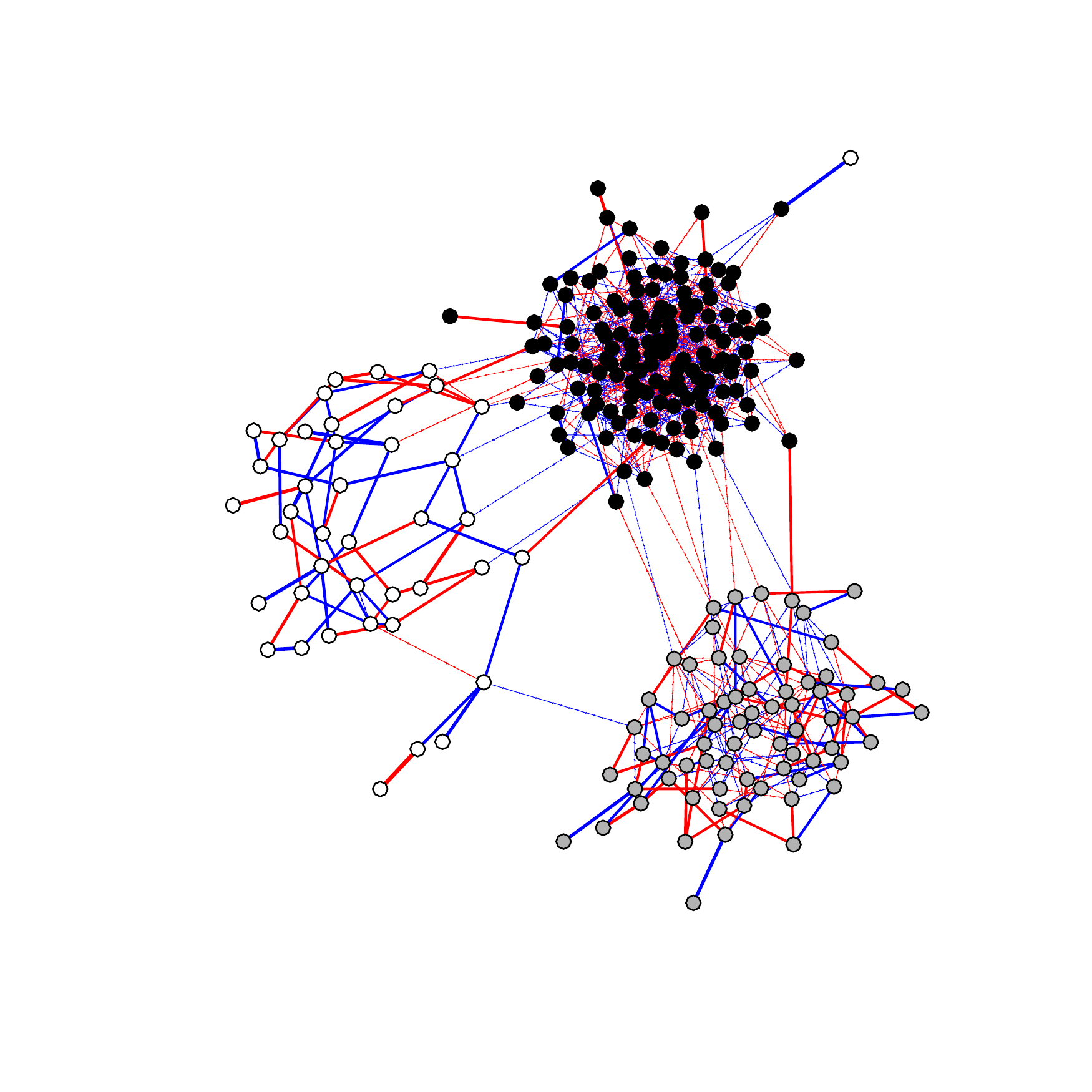} \\
    (a) & (b) & (c) 
  \end{tabular}
  \caption{Simulation of  the structured sparse  concentration matrix.
    Adjacency  matrix   without  (a)  and   with  (b) rows and columns
    reorganized   according  the affiliation  structure  and  corresponding graph
    (c). }
  \label{fig:sampled_graph}
\end{figure}

We simulate sparse graphs with $p=200$  and $n$ from 100 to 2000  
($n/p \in \{1/2,2,3,6,10\}$). We use  
a probability  of intra-cluster connection of  $0.125$, a probability
of  inter-cluster connection  of  $0.0025$, $Q=3$ groups and equal
group proportions $\alpha_i=1/3$. With  these settings,  the
theoretical expected number of edges is about 862 and the total number
of  potential   edges  is   19900.   A  sample   graph  is   given  in
Figure~\ref{fig:sampled_graph}. The running times of \texttt{GLasso} and \texttt{SIMoNe} are
of the same order. For the settings described above the running time
varies from a few seconds to a few minutes, according to the penalty parameter.
\\

We focus the  experiments on the ability to  recover existing edges of
the network, that is the  nonzero entries of the concentration matrix.
This is  a binary decision  problem where the compared  algorithms are
considered as  classifiers. The decision  made by a  binary classifier
can  be summarized using  four numbers:  True Positives  ($TP$), False
Positive ($FP$), True Negatives ($TN$) and False Negatives ($FN$).  We
have chosen to draw precision/recall  curves to display this information and
compare how well the methods perform (Figure \ref{fig:prcurves}).

Precision ($TP/(TP+FP)$) is the ratio of the number of true nonzero
elements to the total number of nonzero elements in the estimated
concentration matrix $\widehat{\mathbf{K}}$.  Recall that ($TP/(TP+FN)$) is
the ratio of true nonzero elements in
$\widehat{\mathbf{K}}$ to all nonzero entries of the real
concentration matrix $\mathbf{K}$.  In a sparse context where the
number of actual positives ($TP+FN$) is small compared to the number
of actual negatives ($FP+TN$), precision/recall curves give a more
informative picture of an algorithm's performance than classical
Receiver Operator Characteristic (ROC) curves. Indeed, ROC curves plot
the False Positive Rate ($FPR=FP/(FP+TN)$) against the True Positive
Rate ($TPR=TP/(TP+FN)$). When the number of total positives is small
compared to the number of total negatives, small variations of $FP$ and
$TP$ will result in small variations of $FPR$ and large variations of
$TPR$, which is not relevant for comparing performances.  In a
statistical framework, the recall is equivalent to the power and the precision
is equivalent to one minus the False Discovery Proportion.

Additionally  to   the  \texttt{GLasso}  \citep{2007_BS_Friedman}  and
\texttt{GeneNet}  \citep{2005_SAGMB_Schafer}  we  consider  two  other
 procedures:
\begin{itemize}
\item
When $n$ is greater than $p$, a straightforward way to obtain
an estimate of the inverse covariance matrix is to invert the
empirical covariance matrix. Although this approach is unlikely to perform well
in a selection context (since it is designed for
estimation purposes), it is worth comparing it to its competitors in order to
assess the scale of improvement. 
We call this procedure \texttt{InvCor}.
\item 
When the latent structure $\mathbf{Z}$ of the concentration matrix is
known, our method can be applied without its E-step and
produce a relevant selection of the nonzero entries of the concentration matrix.
This approach represents the upper limit of our method,
since it makes use of an usually unavailable source of
information. This procedure is denoted \texttt{perfect SIMoNe}.

In some problems  the latent structure of the graph is partially
known and  this information can be  used in the E-step  to improve the
estimation of the latent structure. For example, when inferring gene
regulation networks,  a subset of identified genes may be known to belong
to the same functional module. 
\end{itemize}

The approach of \cite{2006_AS_Meinshausen} was also tested. The
principle of this approach, and the performances obtained are close to those of \texttt{GLasso}, but
it was always slightly outperformed. We have therefore decided, for the sake of
brevity, to report only the four previously described procedures.

For    the   methods    based   on    penalization   (\texttt{GLasso},
\texttt{SIMoNe}  and  \texttt{Perfect SIMoNe}),   the
precision/recall   curves  are   plotted  by   varying   the  penalty
parameter (namely $1/\lambda_{\text{in}}$ in our case).  The penalty  parameter  varies  from close  to  zero to  a
maximum value  which forces all off-diagonal elements of
$\widehat{\mathbf{K}}$ to be null (see
Appendix   \ref{sec:bound}). The   \texttt{GeneNet}   and
\texttt{InvCor}  methods  are  plotted  by  sorting  the  elements  of
$\widehat{\mathbf{K}}$  according to  their  absolute   values,  and   choosing  different
thresholds to find nonzero entries.

%% Comments 
Even when $n$ is really greater than $p$ (Figures \ref{fig:prcurves}
(a-b)) \texttt{Invcor} is always dominated by the other methods from a
selection point of view. This simple check shows that even in a favorable
context with abundant data, penalization procedures improve the selection 
of nonzero entries of the concentration matrix, in comparison with methods
based on estimation of these entries.

Although \texttt{ GeneNet} and \texttt{GLasso} can provide different
results on a given run, both methods perform similarly on average (50
runs for our experiment). The only parameter we change in this
experimental setting is the $n/p$ ratio. 
%Other authors have observed
%different behaviour when simulating their data with a different model.

\texttt{Perfect SIMoNe}'s curves dominate all other curves for any $n/p$
ratio. This clearly shows that the knowledge of the  structure
provides a valuable information for selecting the nonzero entries of
the concentration matrix. When the structure is hidden, the main
problem of our approach is then to find a reliable estimate of this
structure from the initial data.  

\texttt{Perfect SIMoNe} and  \texttt{SIMoNe} perform
 equivalently when $n=10p$ and when the ratio $n/p$ decreases,   
\texttt{Perfect SIMoNe} tends to outperform  \texttt{SIMoNe}
more clearly. This means that \texttt{SIMoNe} is able to recover the
latent structure when there is enough data, but does not find a
substantial structure when $n$ drops below $p$.

When $p>n$,  the empirical covariance matrix ceases to be
invertible. Thus, Figures \ref{fig:prcurves}
(e-f) do not display the \texttt{InvCor} results.  Although it is
possible to show that both \texttt{GLasso} and \texttt{SIMoNe} increase the number of
inferred true nonzero elements with the number of iterations in all settings,
precision/recall curves show the relative poor performances for all tested algorithms
when $p \geq n$.

Notice that when $p>n$, the estimated latent structure is not
reliable. Nevertheless, the performance of \texttt{SIMoNe} remains
comparable to that of \texttt{GLasso}. We can therefore see that assuming the existence of a
latent structure when there is none does not impair the selection
of nonzero entries of the matrix $\mathbf{K}$.

\begin{figure}[htbp]
  \begin{tabular}{@{}c@{}c@{}}
        \includegraphics[width=.5\textwidth]{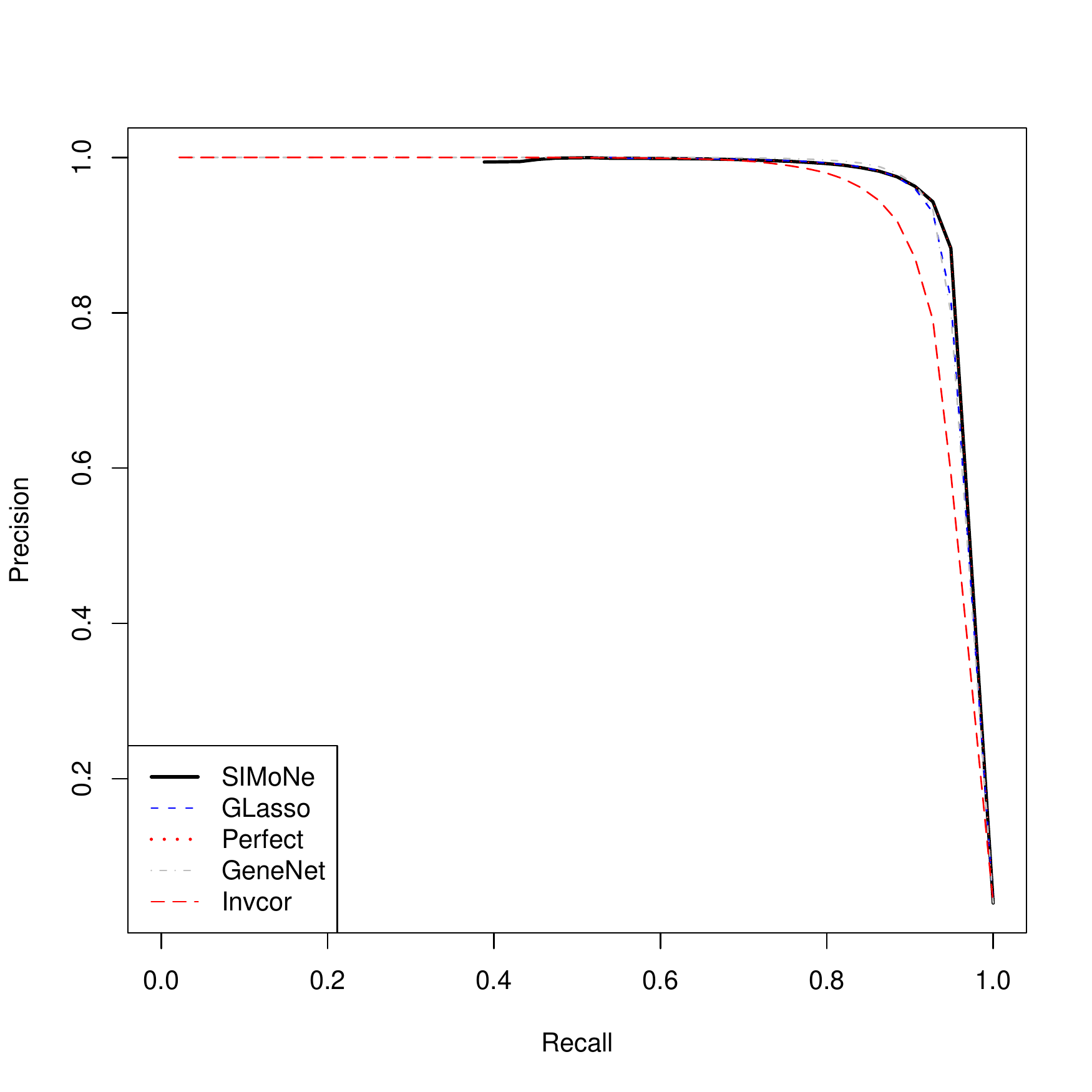}
  &     \includegraphics[width=.5\textwidth]{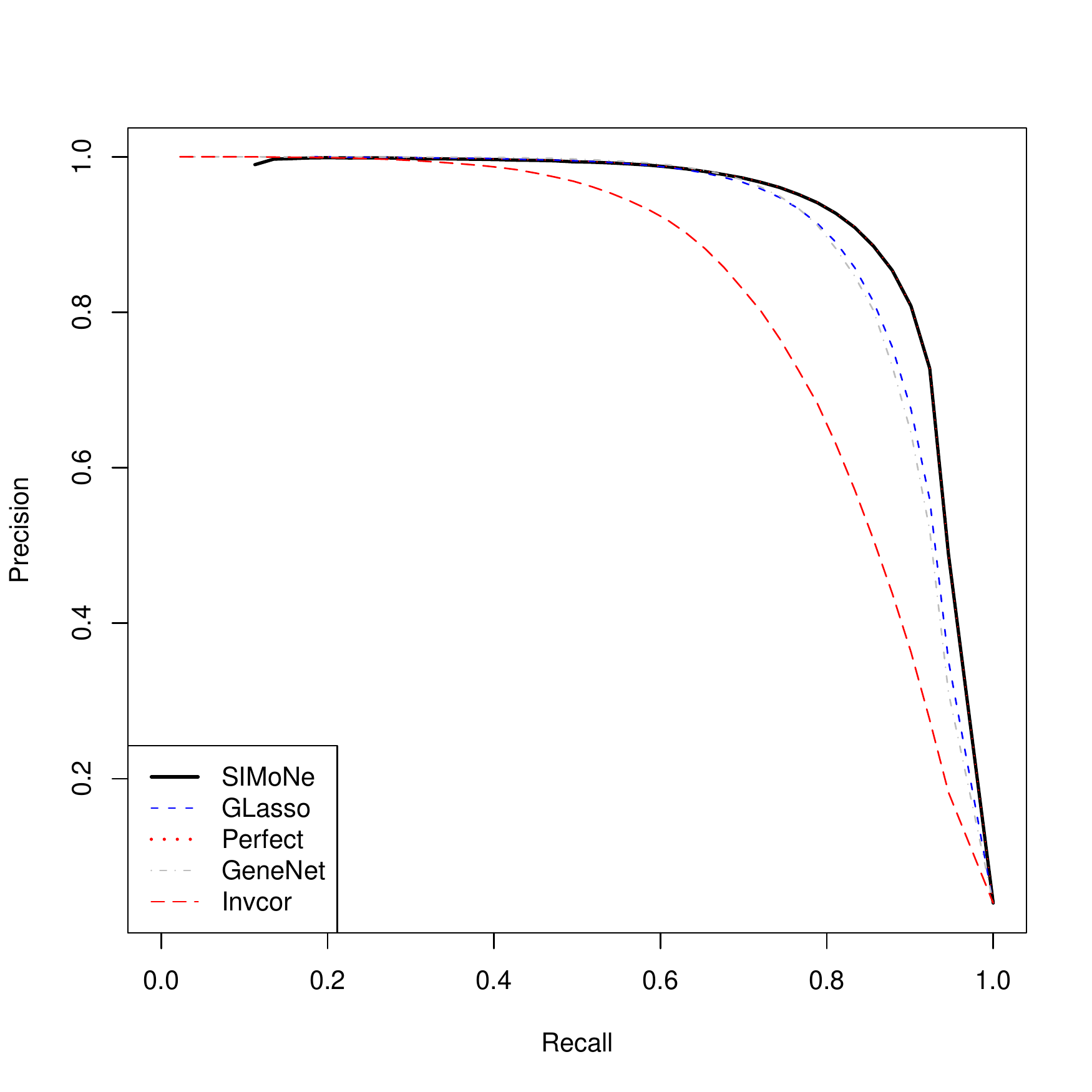} \\
(a) $n=10p$& (b) $n=6p$ \\
        \includegraphics[width=.5\textwidth]{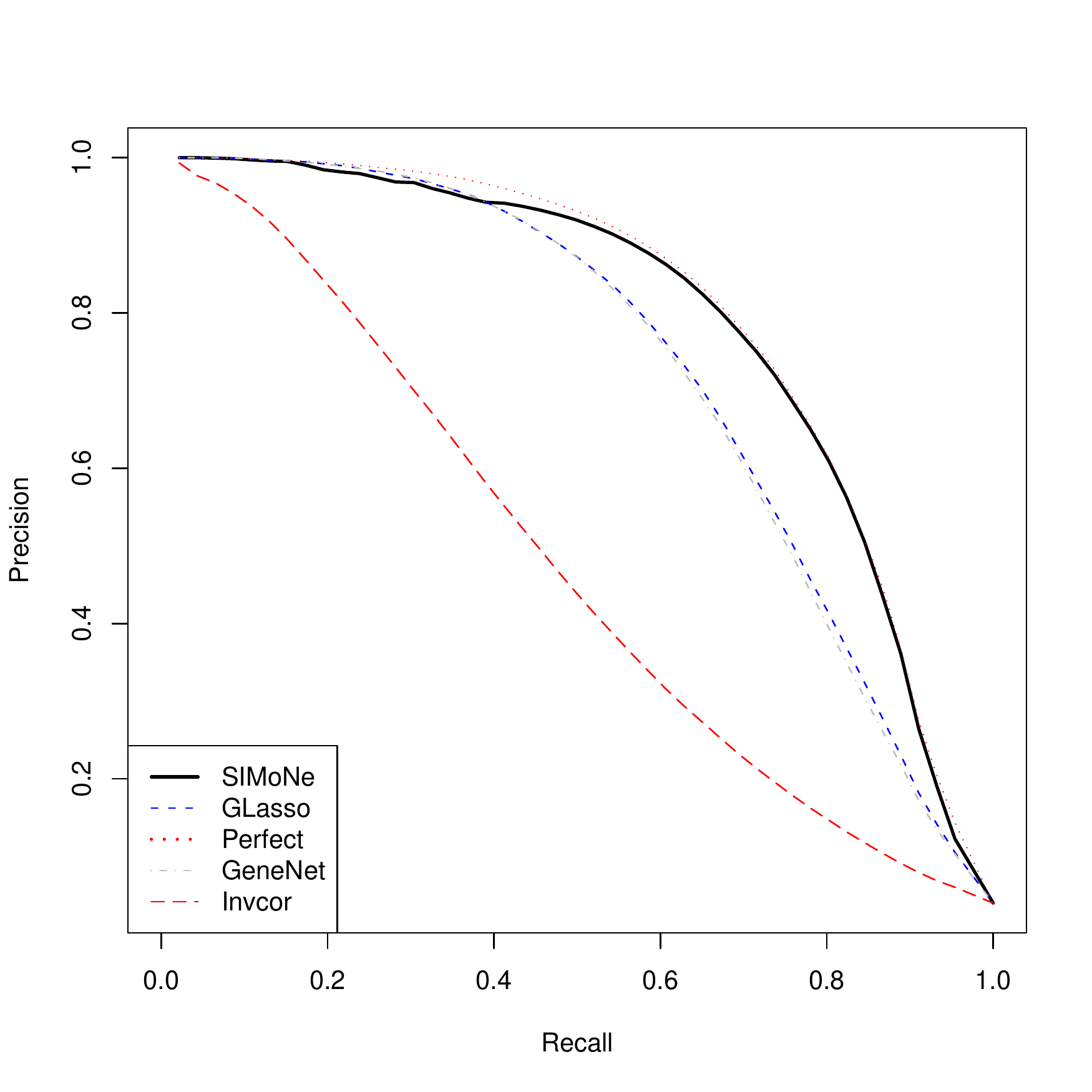}
  &     \includegraphics[width=.5\textwidth]{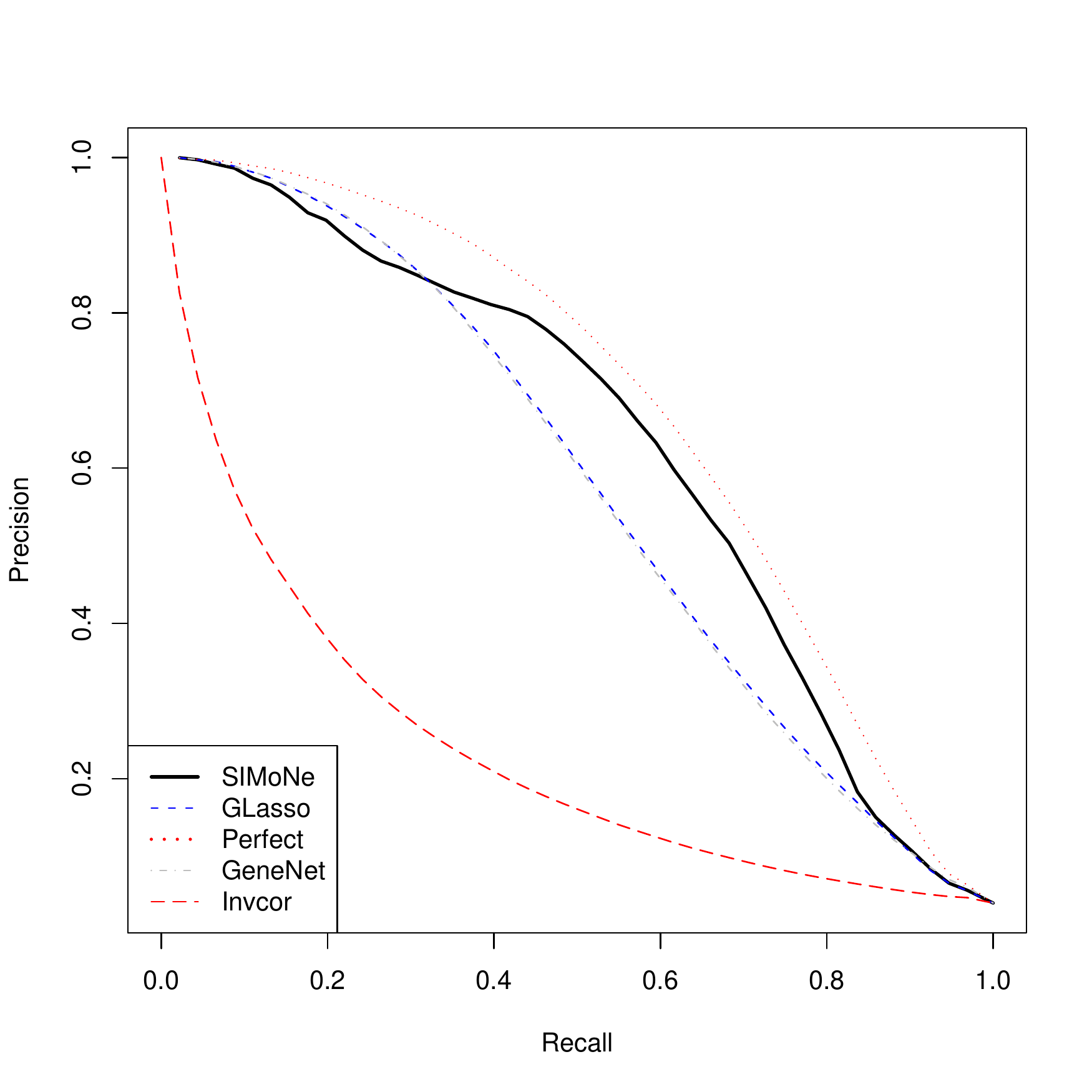} \\
(c) $n=3p$& $n=2p$ (d) \\
        \includegraphics[width=.5\textwidth]{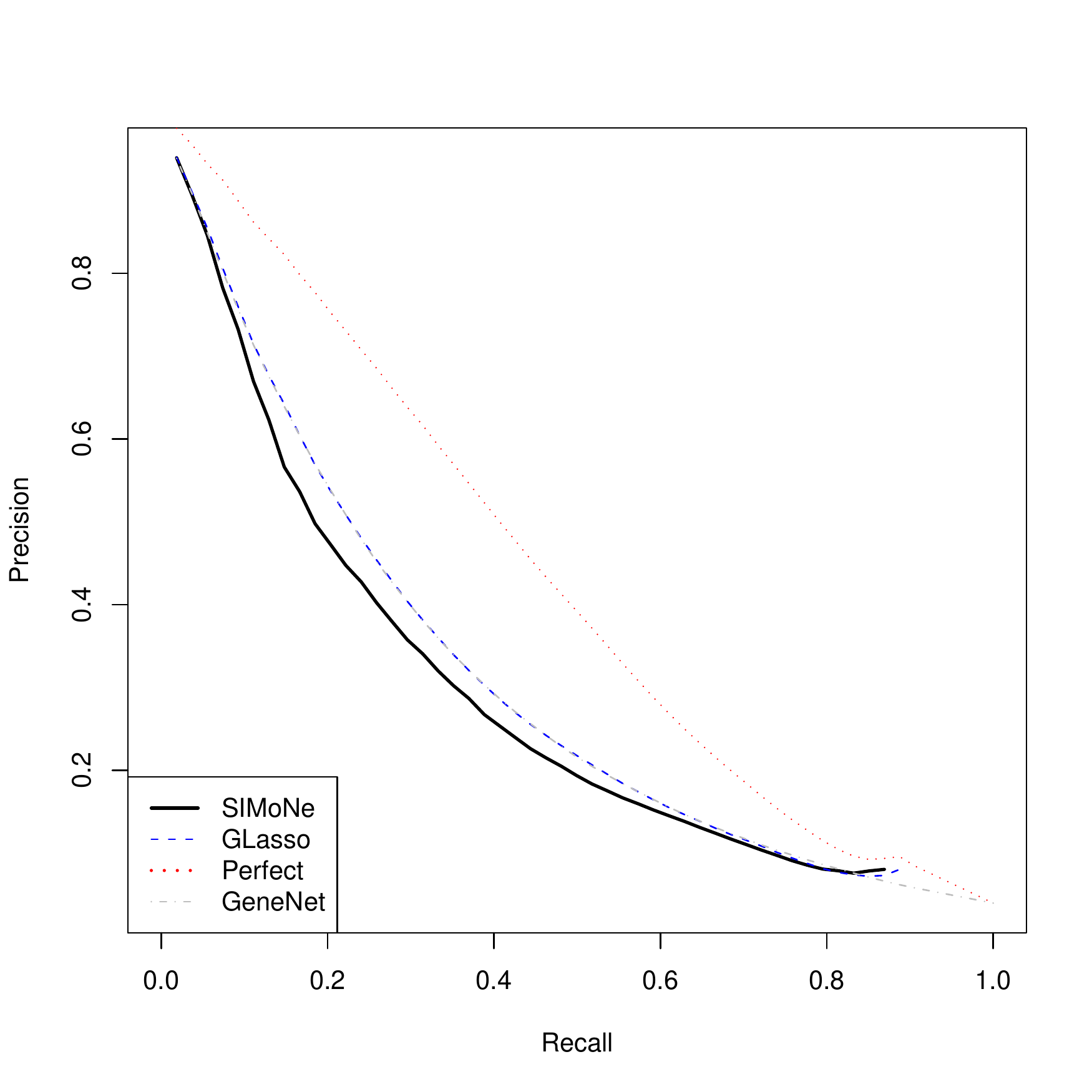}
  &     \includegraphics[width=.5\textwidth]{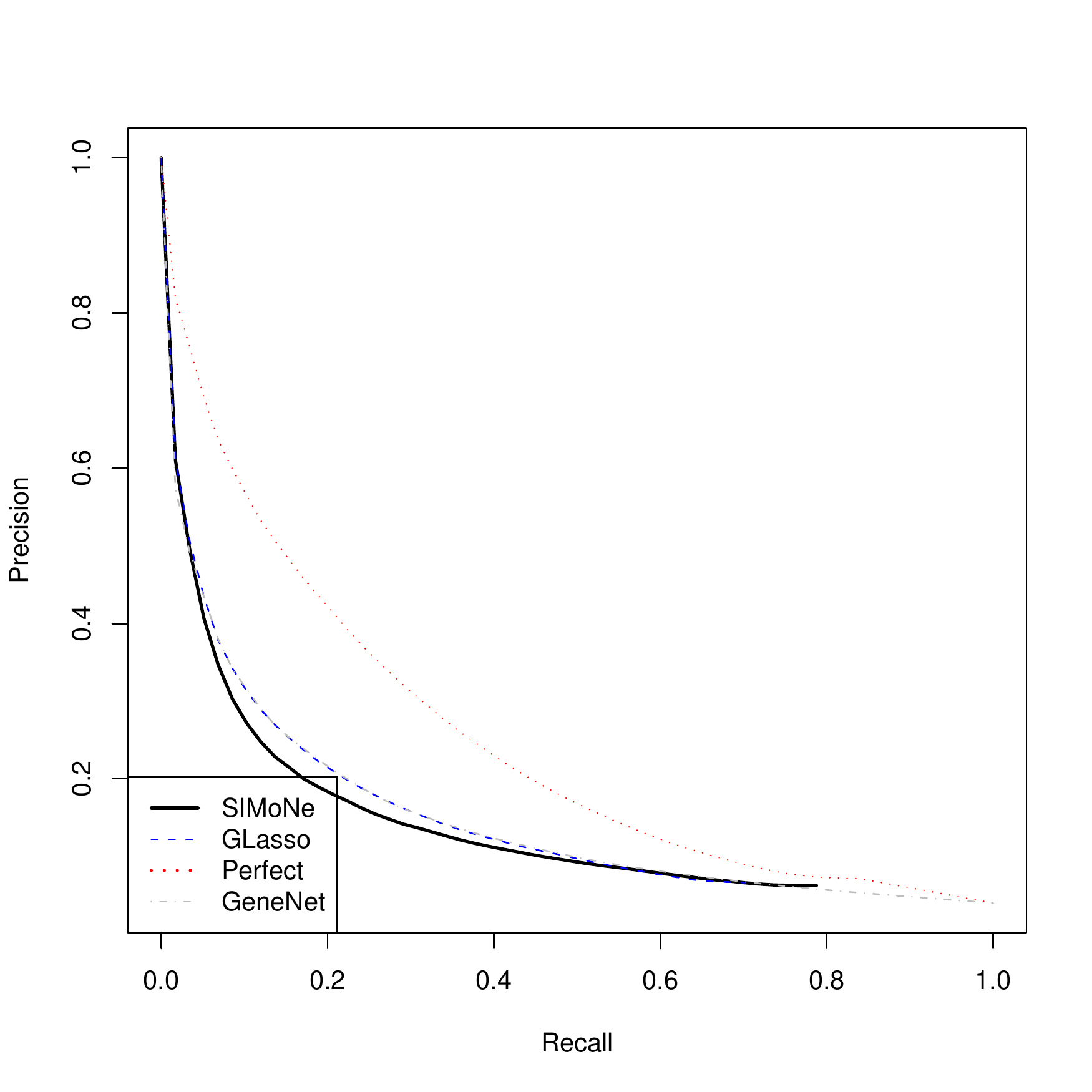} \\
(e) $n=p$& (f) $n=\frac{p}{2}$\\
  \end{tabular}
  \caption{Precision/recall curves comparing the performance of
    GeneNet,
 GLasso, SIMoNe and perfect SIMoNe, when inferring the structure of a
simulated graph with $p=200$ variables.\label{fig:prcurves}}
\end{figure}

\subsection{Breast Cancer data}
%Biological regulatory networks are known to be organized in
%non-homogeneous structures.

We tested our algorithm on a gene expression data set provided by
\citet{2006_JCO_hess} and concerning $133$ patients
with stage $I-III$ breast cancer.  The patients were treated with
chemotherapy prior to surgery. Patient response to the treatment is classified as either a pathologic complete response (pCR) or a 
residual disease (not-pCR).
\citet{2006_JCO_hess} and \citet{2008_BMC_Natowicz} developed and
tested a multigene predictor for treatment response on this data set. They focused on a set of 
%30 probes (
26 genes having a  high  predictive value (see
Table \ref{tab:genes_reference}). We
thus consider a total of  $n=133$ cases containing $p=26$ gene
expression levels.

When dealing with gene regulatory networks, we typically observe $n$
independent microarray experiments, each giving the expression
levels of the same $p$ genes. If the same experimental conditions
are used for all microarrays, these may be considered as a sample of
the same experiment.  In the application in question, cases from the
pCR class (34 cases) and from the not-pCR class (99 cases) clearly do
not have the same distribution. We apply our algorithm on each
class of patients. Two distinct gene regulatory networks are thus
inferred.\\

%% Comments

Figure \ref{fig:cancer} plots the resulting networks obtained for three
different penalizations. The penalization parameters  were heuristically chosen from the number of expected nonzero
entries. We used $Q=2$ latent clusters, and it is interesting to note
that when assuming more than two clusters, the algorithm systematically
produces exactly two non-empty clusters.  

The inferred networks exhibit very different structures
according to the  class of patients. This in itself is
interesting and suggests that  gene regulation differs with respect
to the presence or absence of a pCR.   

The network obtained with not-pCR cases displays a two-star
pattern. Each star connects to a unique gene, either SCUBE2 or IGFBP4.
Almost all the most significant connections imply SCUBE2. This star
pattern suggests that further studies of this particular gene would be
of interest for understanding residual disease.

The network estimated with the pCR cases has a different two-cluster 
structure. In particular, it groups  IGFBP4 and  SCUBE2 in the same
cluster with a direct significant link. This again indicates a
completely different relationship between the genes in pCR versus non-pCR.

\begin{figure}[htbp]
  \centering
  \begin{tabular}{|c | c c|}
 \hline
    & \textbf{not-pCR} & \textbf{pCR} \\ \hline
     \rotatebox{90}{\hspace{2cm }Low penalty}&\includegraphics[width=.475\textwidth]{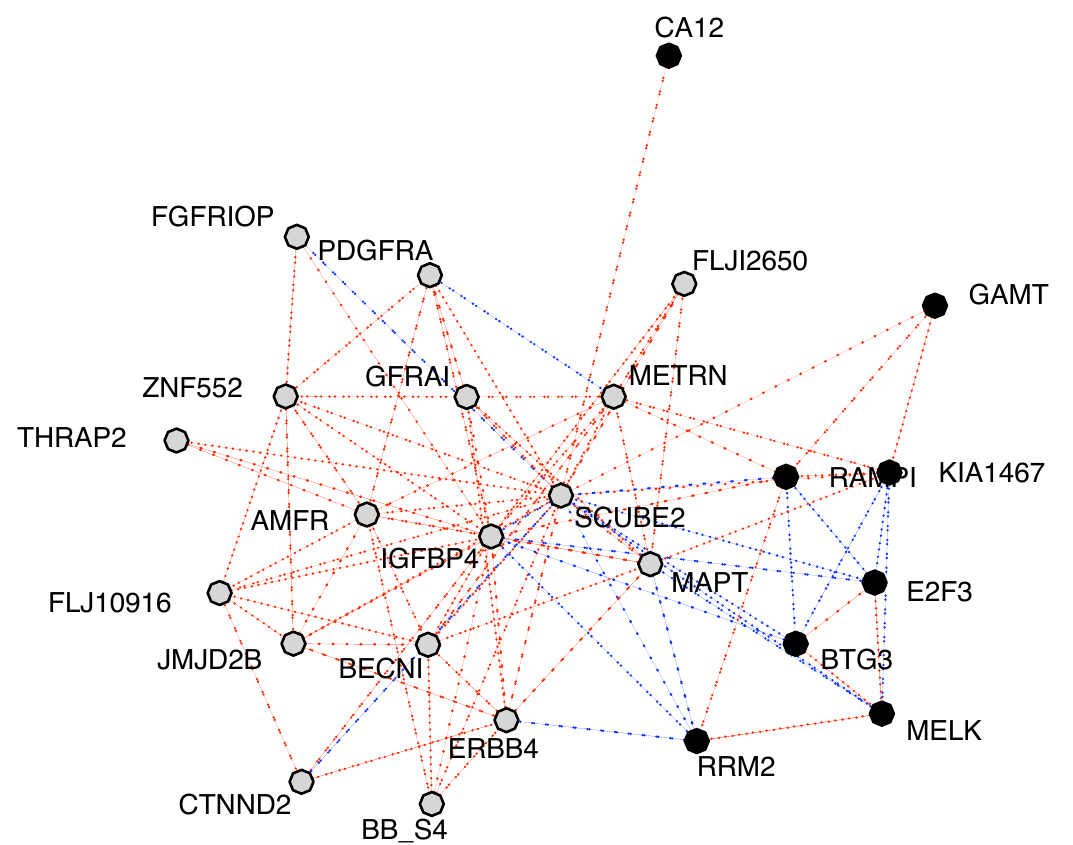}& 
    \includegraphics[width=.475\textwidth]{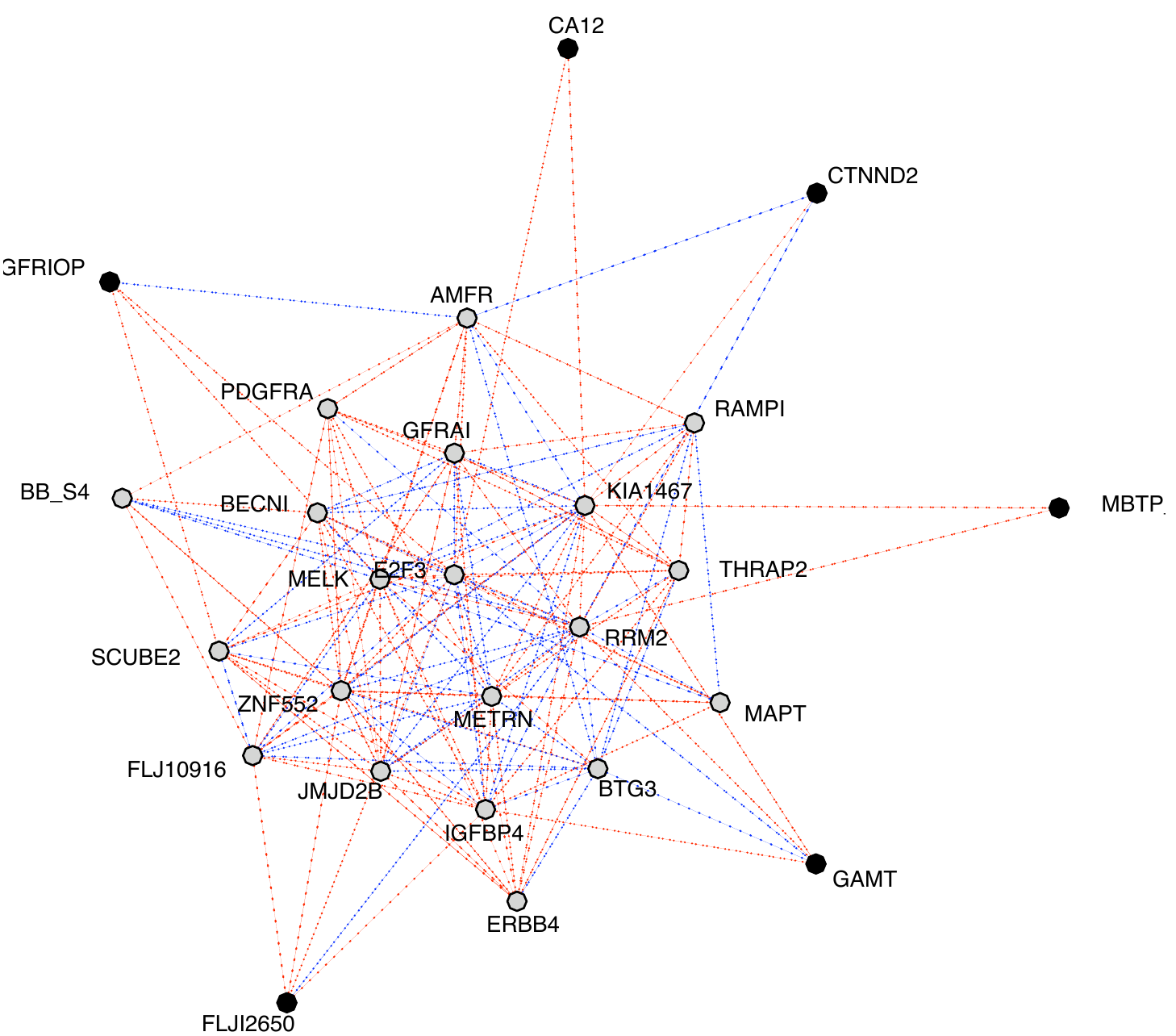} \\  \hline
     \rotatebox{90}{\hspace{2cm } Medium penalty}&\includegraphics[width=.475\textwidth]{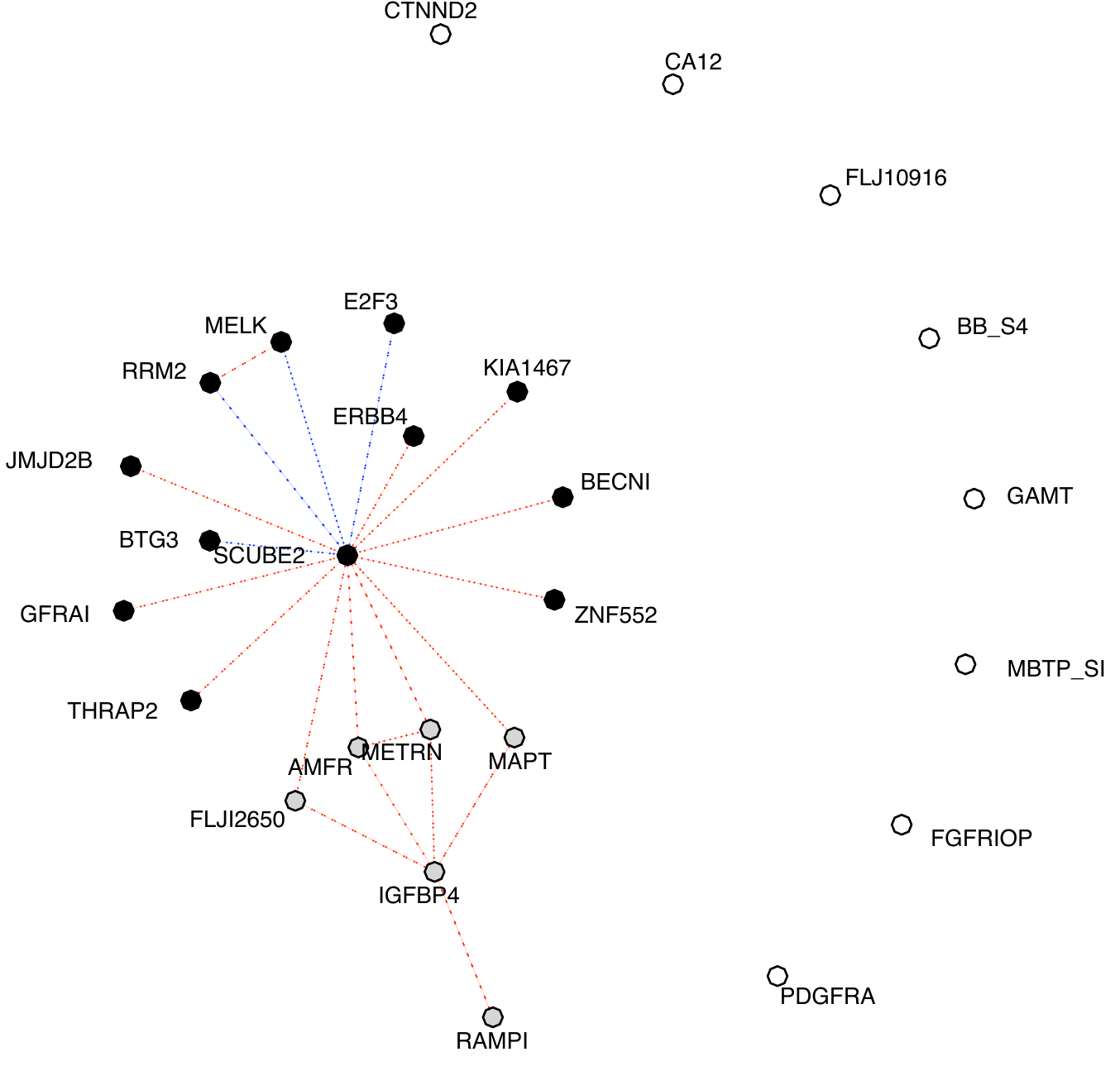}& 
    \includegraphics[width=.475\textwidth]{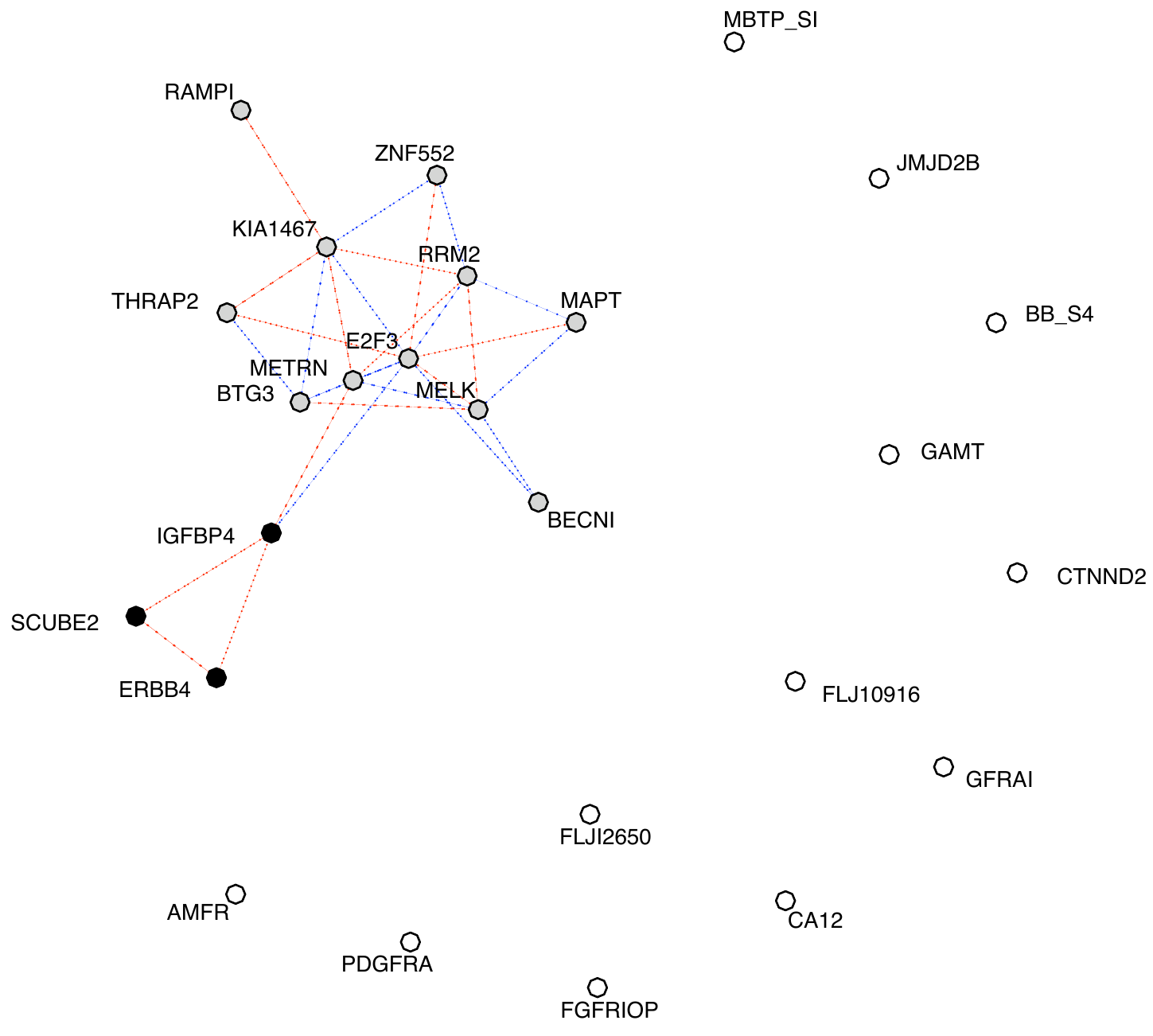} \\  \hline
     \rotatebox{90}{\hspace{2cm } High penalty}&\includegraphics[width=.475\textwidth]{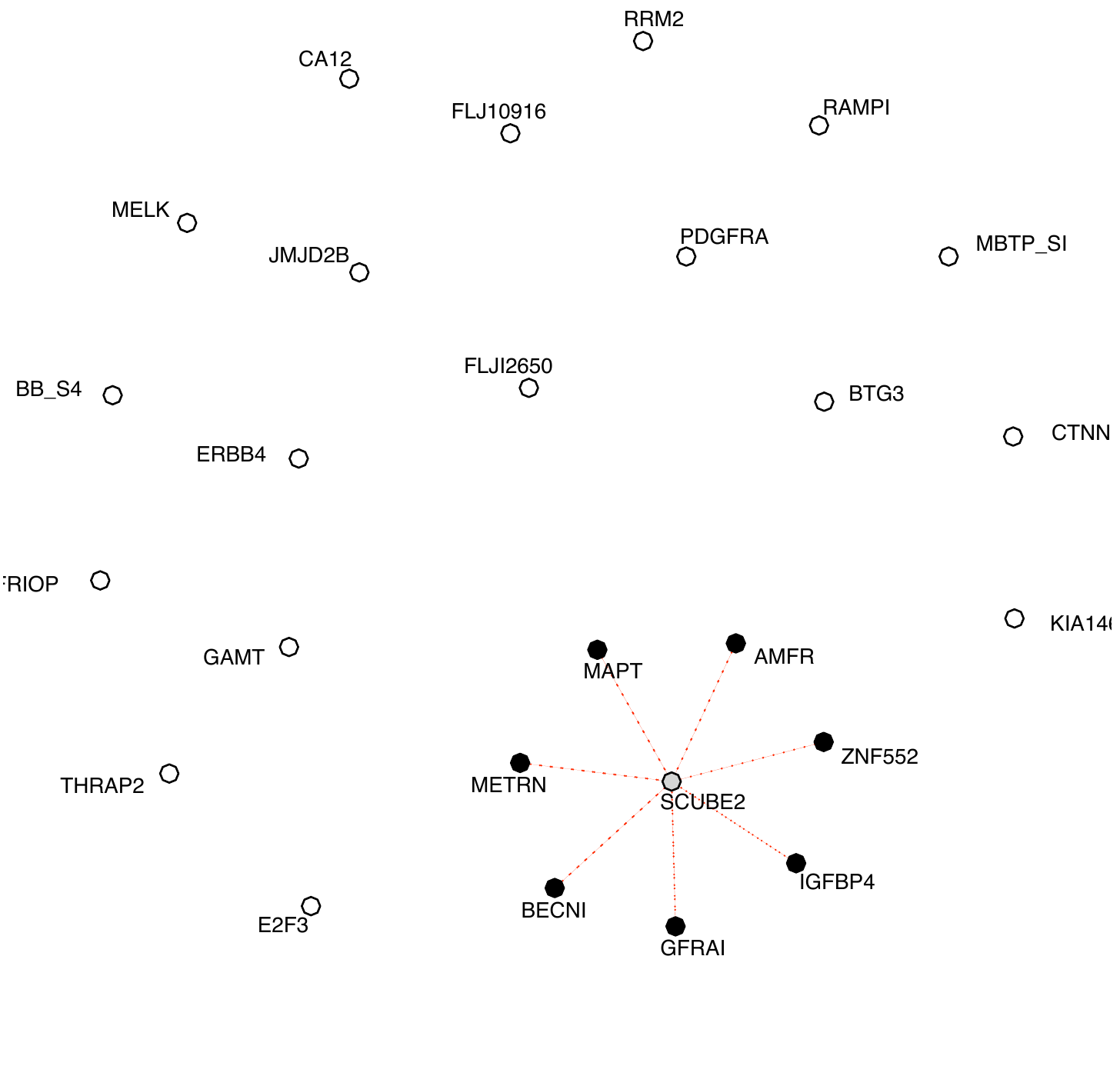}& 
    \includegraphics[width=.475\textwidth]{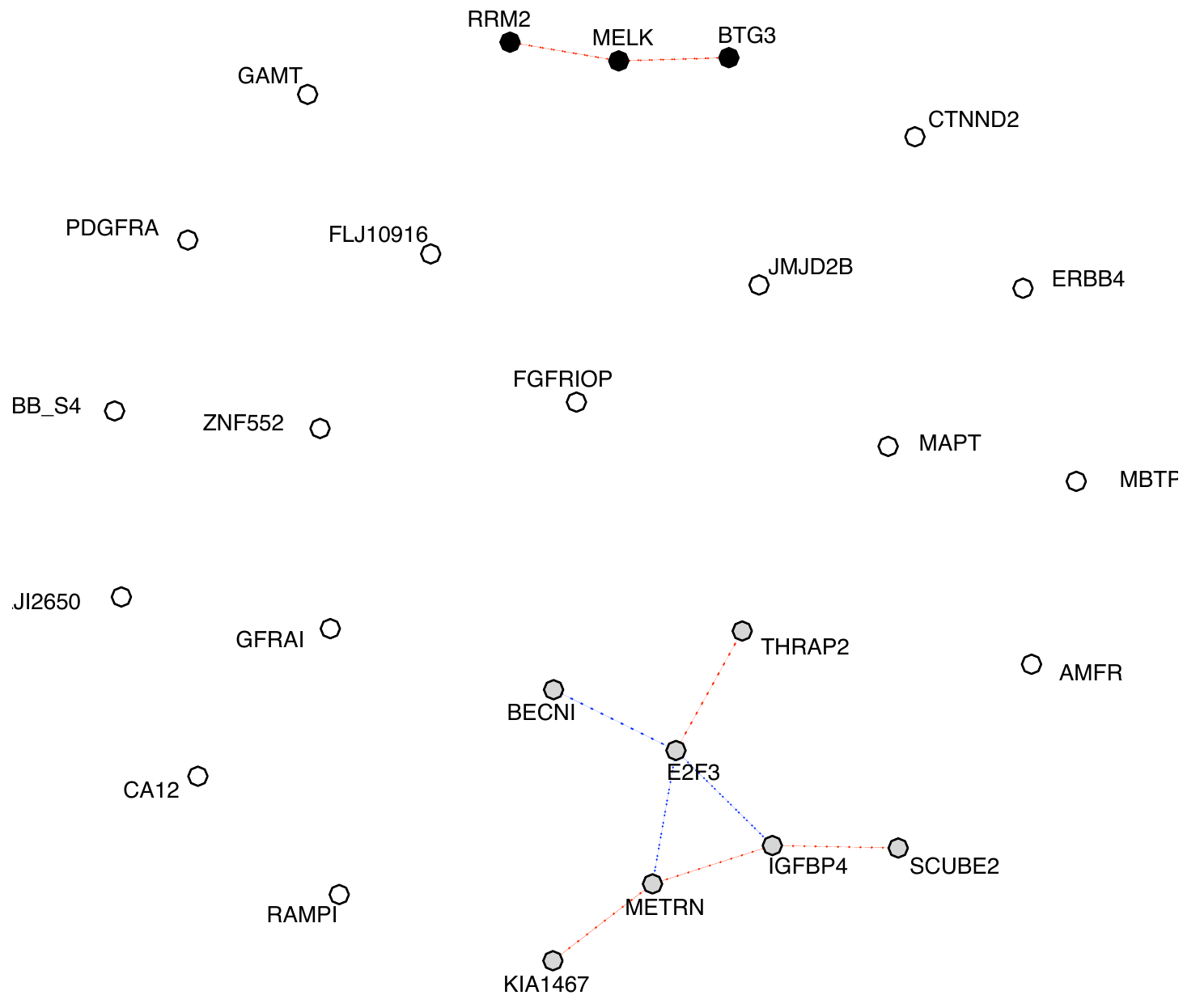} \\  \hline
  \end{tabular}
  \caption{Inferred graphs for three different  penalization's levels.} \label{fig:cancer}
\end{figure}

\begin{table}[htbp!]
  \centering
  \begin{tabular}{ll}
    \hline
    Gene symbol & Gene name \\ 
    \hline
    \hline
    MAPT & Microtubule-associated protein \\
    \hline
    BBS4  & Bardet-Biedl syndrome 4 \\
    \hline
    THRAP2  & Thyroid hormone receptor associated protein 2 \\
    \hline
    MBTP-S1  & Hypothetical protein \\
    \hline
    PDGFRA  & Human clone 23,948 mRNA sequence \\
    \hline
    ZNF552  & Zinc finger protein 552 \\
    \hline
    RAMP1  & Receptor (calcitonin) activity modifying protein 1 \\
    \hline
    BECN1  & Beclin 1 (coiled-coil, myosin-likeBCL2 interacting protein) \\
    \hline
    BTG3  & BTG family, member 3 \\
    \hline
    SCUBE2  & Signal peptide, CUB domain,EGF-like 2 \\
    \hline
    MELK  & Maternal embryonic leucine zipper kinase \\
    \hline
    AMFR  & Autocrine motility factor receptor \\
    \hline
    CTNND2  & Catenin, delta 2  \\
    \hline
    GAMT  & Guanidinoacetate N-methyl transferase \\
    \hline
    CA12  & Carbonic anhydrase XII \\
    \hline
    FGFR1OP  & FGFR1 oncogene partner \\
    \hline
    KIAA1467  & KIAA1467 protein \\
    \hline
    MTRN  & Meteorin, glial cell differentiation regulator \\
    \hline
    FLJ10916  & Hypothetical protein FLJ10916 \\
    \hline
    E2F3  & E2F transcription factor 3 \\
    \hline
    ERBB4  & V-erb-a erythroblastic leukemiaviral oncogene homolog 4(avian) \\
    \hline
    JMJD2B  & Jumonji domain containing 2B \\
    \hline
    RRM2  & Ribonucleotide reductase M2polypeptide \\
    \hline
    FLJ12650  & Hypothetical protein FLJ12650 \\
    \hline
    GFRA1  & GDNF family receptor 1 \\
    \hline
    IGFBP4  & Insulin-like growth factor binding protein 4 \\
  \end{tabular}
  \caption{The key genes that composed the inferred networks.}
  \label{tab:genes_reference}
\end{table}

%%% Local Variables:
%%% mode: latex
%%% TeX-master: "CJC.tex"
%%% End:

%% file: appendix.tex
\section{Appendix section}

\subsection{Proof of the equivalence between the constraints}\label{ap:proof_norm}

When $\left\|\mathbf{U}_{q\ell}\right\|_\infty\leq 1$,  we have
for each couple $(i,j) \in \mathcal{P}^2$,
\begin{equation*}
  \left| \left(\boldsymbol\Sigma  - \mathbf{S}\right)_{ij}\right| = \frac{2}{n}
  \sum_{q,\ell} \left|(\mathbf{U}_{q\ell})_{ij} \cdot
    (\mathbf{T}_{q\ell})_{ij}\right| 
  \leq                     \frac{2}{n}                    \sum_{q,\ell}
  T_{q\ell;ij} = P_{\boldsymbol\tau_i \boldsymbol \tau_j}.
\end{equation*}
Thus   $\left\|\mathbf{U}_{q\ell}\right\|_\infty\leq   1   \Rightarrow
\left\|    (\boldsymbol\Sigma    -    \mathbf{S}\right)    \cdot\slash
\mathbf{P}_{\boldsymbol \tau}\|_\infty \leq 1$.

On  the   other  hand,  assume  that   $\left\|  (\boldsymbol\Sigma  -
  \mathbf{S})\cdot\slash                        \mathbf{P}_{\boldsymbol
    \tau}\right\|_\infty\leq   1$,   that  is,   for   all  $i,j   \in
\mathcal{P}$, we have
\begin{equation*}
  -P_{\boldsymbol\tau_i \boldsymbol \tau_j}\leq    \left(\boldsymbol\Sigma   -\mathbf{S}\right)_{ij}\leq P_{\boldsymbol\tau_i \boldsymbol \tau_j}.
\end{equation*}
This also means that there exists some $\delta_{ij}\in[0,1]$ such that
\begin{align*}
  \left(\boldsymbol\Sigma       -\mathbf{S}\right)_{ij}       &      =
  \delta_{ij}P_{\boldsymbol\tau_i      \boldsymbol      \tau_j}      +
  (1-\delta_{ij})(-P_{\boldsymbol\tau_i    \boldsymbol    \tau_j})   =
  \frac{2}{n}\sum_{q,\ell}(2\delta_{ij}-1) T_{q\ell;ij}.
\end{align*}
We   choose  $\mathbf{U}_{ql}$   such   that  $(\mathbf{U}_{ql})_{ij}=
(2\delta_{ij}-1)$   for  all   $q,\ell\in\mathcal{Q}$.    Then,  since
$\delta_{ij}\in[0,1]$, we have
\begin{equation*}
  -1\leq (\mathbf{U}_{q\ell})_{ij} \leq 1, \qquad \forall i,j \in \mathcal{P},
\end{equation*}
which proves that $\left\| (\boldsymbol\Sigma - \mathbf{S})\cdot\slash
  \mathbf{P}_{\boldsymbol          \tau}\right\|_\infty\leq          1
\Rightarrow\|\mathbf{U}_{q\ell} \|_\infty\leq 1$.

\subsection{Fixed-point study}\label{app:point_fixe}
Let us first  introduce some notation. For any  $i,j \in \mathcal{P}$
and any $q,\ell \in \mathcal{Q}$, consider the random variables
\begin{equation*}
   L_{ijq\ell} = \frac  {|K_{ij}|}{\lambda_{q\ell}} +\log 2\lambda_{q\ell} . 
 \end{equation*}
 Let $u  : \mathbb{R}^{pQ} \rightarrow \mathbb{R}^{pQ}$  be defined by
 its     coordinate     functions    $u=(u_{iq})_{i\in\mathcal{P},q\in
   \mathcal{Q}}$ in the following way
 \begin{multline*}
   \forall a=(a_{iq})_{i\in\mathcal{P},q\in \mathcal{Q}} \in
   \mathbb{R}^{pQ}, \\
   u_{iq}(a)  =   \alpha_q  \exp\Big\{  -\sum_{j\neq   i}  \sum_{\ell}
   a_{j\ell} L_{ijq\ell} \Big\} \\ = \alpha_q \exp\Big\{ -\sum_{j\neq i}
   \sum_{\ell}  a_{j\ell}   \left(  \frac  {|K_{ij}|}{\lambda_{q\ell}}
     +\log 2\lambda_{q\ell}\right) \Big\} ,
\end{multline*}
and   let   $   g=   (g_{iq})_{i\in\mathcal{P},q\in   \mathcal{Q}}   :
\mathbb{R}^{pQ} \rightarrow \mathbb{R}^{pQ}$ satisfy
\begin{equation*}
 \forall a
%=(a_{iq})_{i\in\mathcal{P},q\in \mathcal{Q}} 
\in \mathbb{R}^{pQ}, \quad g_{iq}(a) =
 \frac{u_{iq}(a)}{\sum_{\ell}u_{i\ell}(a)}.   
 \end{equation*}
 According to  Proposition~\ref{prop:pointfixe}, the optimal parameter
 $\widehat{\boldsymbol \tau}$ is a fixed-point of $g$.

Now, let 
\begin{equation*}
  \Theta=\Big\{ a= (a_{iq})_{i\in\mathcal{P},q\in \mathcal{Q}} 
\in \mathbb{R}^{pQ} ; \forall i \in \mathcal{P}, q\in
\mathcal{Q}, a_{iq}\in [0,1] \text{ and } \sum_{q}
a_{iq}= 1 \Big\} .
\end{equation*}
We wish to study the fixed-points of $g$ in
$ \Theta$. First, let us note that as $\Theta$ is a compact state
space and as the function $g$ satisfies $g :  \Theta \rightarrow  \Theta$
and is continuous, the existence of a fixed-point of $g$ follows from
Brouwer's Theorem. 

We now restrict our attention to a smaller set than the whole
state space $\Theta$. For any $\varepsilon >0$, let 
\begin{equation*}
  \Theta_{\varepsilon}=\Big\{ a \in \Theta, \forall i \in \mathcal{P}, q\in
\mathcal{Q}, a_{iq}\in [\varepsilon, 1-\varepsilon] \Big\} . 
\end{equation*}
Note that  we do  not claim that  $g :  \Theta_\varepsilon \rightarrow
\Theta_\varepsilon $.  However, the existence of  a fixed-point of  $g$ is
ensured  in $\Theta$  and  if  we assume  $\alpha_q>0$  for any  $q\in
\mathcal{Q}$  (which  is a  reasonable  assumption  if  the number  of
classes $Q$ is not too large),  it can easily  be seen that any fixed-point
satisfies  $a_{iq}>0$,   for  any  $i\in  \mathcal{P}$   and  any  $q\in
\mathcal{Q}$. Thus for sufficiently small $\varepsilon>0$, the fixed-points
of $g$ belong to $\Theta_\varepsilon$.

In order to study the behaviour of $g$ in the vicinity of a fixed-point,
we need to look at some kind of contraction property for $g$. To this
end we introduce a distance $d$ on $\Theta_\varepsilon$ that will
make use of the form of the state space
$\Theta_\varepsilon$. For all $a,b \in \Theta_\varepsilon$, denote by 
$a_i=(a_{iq})_{q\in \mathcal{Q}} \in \mathbb{R}^{Q}$ and
$b_i=(b_{iq})_{q\in \mathcal{Q}} \in \mathbb{R}^{Q}$. Moreover, let 
\begin{equation*}
  d(a,b) = \max_{i\in \mathcal{P}} d_0(a_i,b_i) = \max_{i\in
    \mathcal{P}} \log \left( \frac{\max_{q\in \mathcal{Q}}
      a_{iq}/b_{iq}} {\min_{q\in \mathcal{Q}} a_{iq}/b_{iq}}\right) = \max_{i\in
    \mathcal{P}} \max_{q,\ell \in \mathcal{Q}} \log \left(
    \frac{a_{iq}b_{i\ell}}{b_{iq} a_{i\ell}} \right).
\end{equation*}
It  is   well known  that  $d_0$  is  a   distance  in  $[\varepsilon,
1-\varepsilon]^{Q}$, and it is easy to check that the resulting $d$ is also a
distance in $\Theta_\varepsilon$.

Now,  fix   $a,b  \in  \mathbb{R}^{pQ}$  and   consider  the  distance
$d(g(a),g(b))$. It is easily checked that
\begin{equation*}
  d(g(a),g(b)) = \max_{i\in \mathcal{P}} d_0(g_i(a),g_i(b)) =
  \max_{i\in \mathcal{P}} d_0(u_i(a),u_i(b)) = \max_{i\in \mathcal{P}} d_0(\bar{u}_i(a),\bar{u}_i(b)),
\end{equation*}
where       $\bar       u=       (\bar{u}_i)_{i\in       \mathcal{P}}=
(\bar{u}_{iq})_{i\in\mathcal{P},q\in  \mathcal{Q}}$ is defined  in the
following way
 \begin{multline*}
 \forall a=(a_{iq})_{i\in\mathcal{P},q\in \mathcal{Q}} \in
 \mathbb{R}^{pQ}, \\
 \bar{u}_{iq}(a) =  \exp\Big\{
 \sum_{j\neq i} \sum_{\ell} a_{j\ell} L_{ijq\ell} \Big\} 
  = \exp\Big\{ \sum_{j\neq i} \sum_{\ell} a_{j\ell} \left( \frac
  {|K_{ij}|}{\lambda_{q\ell}} +\log 2\lambda_{q\ell}\right) \Big\} .
\end{multline*}

In    the    following,    fix    $\varepsilon>0$   and    $a,b    \in
\Theta_{\varepsilon}$ and denote by
\begin{equation*}
 \forall i \in \mathcal{P}, \quad c_1^i= \min_{q\in \mathcal{Q}}
 \frac{a_{iq}}{b_{iq}} , \quad  c_2^i= \max_{q\in \mathcal{Q}}
 \frac{a_{iq}}{b_{iq}} .
\end{equation*}
With these notations, we have
\begin{equation}
  \label{eq:distab}
  d(a,b)=\max_{i\in \mathcal{P}}d_0(a_i,b_i)=\max_{i\in \mathcal{P}}\log\left( \frac{c_2^i}{c_1^ i}\right).
\end{equation}

We only consider the affiliation model described in
\eqref{eq:affiliation}. Thus, there are only two different values for $\lambda_{q\ell}$, namely $\lambda_{\text{in}}$ and $\lambda_{\text{out}}$ for intra and extra cluster connectivity. 
\begin{lemma}
  If for any $i,j\in\mathcal{P}, i\neq j$ and any $\lambda \in \{\lambda_{\text{in}}, \lambda_{\text{out}}\}$, we have
  \begin{equation}\label{eq:smallroom}
    0 < \frac{|K_{ij}|}{\lambda} +\log 2\lambda < \frac{\varepsilon}{2(p-1)(1+\varepsilon)} \text{ almost surely,}
  \end{equation}
then the function $g$ satisfies a contraction property on $\Theta_{\varepsilon}$. 
\end{lemma}

Before proving the lemma, let us explain the consequences of this result. Consider the function $h_{K}$ defined on $(0,+\infty)$ by 
\begin{equation*}
  h_{K}(\lambda) =  \frac{|K| }{\lambda} +\log 2\lambda.
\end{equation*}
This function first decreases from $+\infty$ to the value $1+\log 2|K|$ on the interval $(0,|K|)$ and then increases from $1+\log 2|K|$ to $+\infty$ on $(|K|,+\infty)$. 

At any step of the algorithm, if the current values $K_{ij}^ {(m)}$ of the concentration matrix are small enough, namely smaller than $1/(2e)\simeq 0.184$ then the functions $h_{K_{ij}^ {(m)}}$ take all the  values between   $1+\log 2|K|<0$ and $+\infty$. Thus,  there is room for choosing  $\lambda_{\text{in}},\lambda_{\text{out}}$ such that \eqref{eq:smallroom} is satisfied. In such a case, the fixed-point we are looking for is unique and the iterative procedure setting $\widehat\tau^{(s+1)}=g(\widehat\tau^{(s)})$ converges.

\begin{proof}
  Using   that  for  any   $j\in  \mathcal{P}$   and  any   $\ell  \in
  \mathcal{Q}$,  we  have  $c_1^j  b_{j\ell}\le  a_{j\ell}  \le  c_2^j
  b_{j\ell}$ and $L_{ijq\ell}>0$, we get
\begin{equation*}
\exp\left( \sum_{j\neq i} c_1^j
    \sum_{\ell} b_{j\ell}L_{ijq\ell}\right) \le \bar{u}_{iq}(a) \le \exp\left( \sum_{j\neq i} c_2^j    \sum_{\ell} b_{j\ell}L_{ijq\ell}\right).
\end{equation*}
Thus, it follows 
\begin{equation}\label{eq:ptfixe_1}
\exp\left( \sum_{j\neq i} (c_1^j-1)
    \sum_{\ell} b_{j\ell}L_{ijq\ell}\right) \le \frac{\bar{u}_{iq}(a)}{\bar{u}_{iq}(b)} \le \exp\left( \sum_{j\neq i} (c_2^j-1)    \sum_{\ell} b_{j\ell}L_{ijq\ell}\right).
\end{equation}
In the case of the affiliation model, for fixed $i,j \in
\mathcal{P}$ and $q\in \mathcal{Q}$, the set of random
variables $\{L_{ijq\ell}\}_{\ell \in \mathcal{Q}}$ is reduced to only
two random values, namely
\begin{equation*}
  L_{ij}^{\text{in}} =  \frac  {|K_{ij}|}{\lambda_{\text{in}}} +\log
  2\lambda_{\text{in}} , \quad   L_{ij}^{\text{out}} =  \frac  {|K_{ij}|}{\lambda_{\text{out}}} +\log  2\lambda_{\text{out}} .
\end{equation*}
For the sake of  simplicity, we assume $Q=2$ groups  (our arguments may be
easily  generalized   to  $3$  groups  or  more).    Now,  denoting  $
L_{ij}^{\text{max}}=  \max(  L_{ij}^{\text{in}}, L_{ij}^{\text{out}})$
and      $     L_{ij}^{\text{min}}=      \min(     L_{ij}^{\text{in}},
L_{ij}^{\text{out}})$,  it  can  easily  be seen  that  (for  $\varepsilon
<1/2$),
\begin{eqnarray*}
  \sup_{b\in \Theta_{\varepsilon}}  \sum_{\ell} b_{j\ell}L_{ijq\ell} &=&
  (1-\varepsilon) L_{ij}^{\text{max}} + \varepsilon
  L_{ij}^{\text{min}} \\
\inf_{b\in \Theta_{\varepsilon}}  \sum_{\ell} b_{j\ell}L_{ijq\ell} &=&
  (1-\varepsilon) L_{ij}^{\text{min}} + \varepsilon
  L_{ij}^{\text{max}}, 
\end{eqnarray*}
almost surely. Note  that if we have $Q\ge  3$ groups, explicit bounds
can  also be  obtained (their  expression is  only slightly more
complicated).  Coming back to \eqref{eq:ptfixe_1}, we get
\begin{multline*}
  \exp\left( \sum_{j\neq i} (c_1^j-1)\{ (1-\varepsilon)
    L_{ij}^{\text{min}} + \varepsilon L_{ij}^{\text{max}}\} \right) \\
  \le \frac{\bar{u}_{iq}(a)}{\bar{u}_{iq}(b)} \le \exp\left( \sum_{j\neq
      i} (c_2^j-1) \{ (1-\varepsilon) L_{ij}^{\text{max}} + \varepsilon
    L_{ij}^{\text{min}}\} \right).
\end{multline*}
This leads to 
\begin{multline*}
  d_0(\bar{u}_i(a),\bar{u}_i(b)) = \log \frac {\max_{q\in \mathcal{Q}}
    \bar{u}_{iq}(a)/\bar{u}_{iq}(b)}{\min_{q\in \mathcal{Q}}
    \bar{u}_{iq}(a)/\bar{u}_{iq}(b)} \\
  \le  \sum_{j\neq
    i} (c_2^j-1) \{ (1-\varepsilon) L_{ij}^{\text{max}} + \varepsilon
  L_{ij}^{\text{min}}\} -\sum_{j\neq
    i} (c_1^j-1) \{ (1-\varepsilon) L_{ij}^{\text{min}} + \varepsilon
  L_{ij}^{\text{max}}\} \\
  \le  \sum_{j\neq
    i}  L_{ij}^{\text{max}}\{c_2^j-1-\varepsilon(c_2^j+c_1^j-2)\}  + 
  L_{ij}^{\text{min}}\{1-c_1^j+\varepsilon(c_2^j+c_1^j-2)\} .
\end{multline*}
Finally,             recall             that            $d(g(a),g(b))=
\max_{i}d_0(\bar{u}_i(a),\bar{u}_i(b))$, leading to
\begin{multline*}
  d(g(a),g(b))           \le          \max_{i\in          \mathcal{P}}
  \left\{     \Big(c_2^i-1-\varepsilon(c_2^i+c_1^i-2)    \Big)    \vee
    \Big(1-c_1^i+\varepsilon(c_2^i+c_1^i-2)\Big) \right\} \\ \times
  \max_{i\in  \mathcal{P}}  \sum_{ j\neq  i}  ( L_{ij}^{\text{max}}  +
  L_{ij}^{\text{min}} ).
\end{multline*}
Now, using the  inverse triangle inequality, and the  fact that $c_1^i
\le 1\le c_2^i$, we get for any $i\in \mathcal{P}$,
\begin{equation*}
|c_2^i+c_1^i -2| = \big|  |c_2^i-1|-|1-c_1^i| \big| \le |c_2^i-c_1^i| = c_2^i-c_1^i.
\end{equation*}
Moreover, we have $0\le c_2^i-1\le c_2^i -c_1^i$ and $0 \le 1-c_1^i\le c_2^i -c_1^i$. This leads to 
\begin{multline}\label{eq:step3}
 d(g(a),g(b))  \le  (1+\varepsilon) \max_{i\in \mathcal{P}} (c_2^i- c_1^i) \times \max_{i\in \mathcal{P}}\sum_{j\neq i} ( L_{ij}^{\text{max}}  + 
  L_{ij}^{\text{min}} ) \\
\le (1+\varepsilon) \max_{i\in \mathcal{P}} (c_2^i- c_1^i) \times 2(p-1)\max_{j\neq i}  L_{ij}^{\text{max}} .
\end{multline}
Since $a$ and $b$ belong to $ \Theta_\varepsilon$, we get that $c_1^i, c_2^i \in [\varepsilon, \varepsilon ^{-1}]$ and thus
\begin{equation*}
  c_2^i- c_1^i = \exp(\log c_2^i)-\exp(\log c_1^i)\le \frac {1} {\varepsilon} \log\left( \frac{c_2^i}{c_1^i}\right).
\end{equation*}
In particular, recalling \eqref{eq:distab}, we have 
\begin{equation*}
  0\le \max_{i\in \mathcal{P}} c_2^i- c_1^i \le \frac 1 \varepsilon d(a,b).
\end{equation*}
Coming back to \eqref{eq:step3}, we get
\begin{equation}\label{eq:contraction}
  d(g(a),g(b))  \le  (1+\varepsilon^{-1}) 2(p-1) \Big(\max_{j\neq i}  L_{ij}^{\text{max}}  \Big) d(a,b).
\end{equation}
Now, under  assumption \eqref{eq:smallroom} the  multiplicative random
factor $(1+\varepsilon^{-1}) 2(p-1) \max_{j\neq i} L_{ij}^{\text{max}}
$ is strictly smaller than $1$.
\end{proof}

\subsection{Proof  of  Lemma~\ref{lem:coord_descent_2} (\textsc{Lasso}
  with pathwise coordinate optimization)}\label{ap:pathwise}

The following is partly  based on \cite{2007_AAS_Friedman}.  There are
various algorithms for solving the \textsc{Lasso} problem.  When there
is just one predictor, the  \textsc{Lasso} solution is simply given by
soft-thresholding \citep{1995_JASA_Donoho}. The  approach used here is
based on iterative soft-thresholding with a ``partial residual'' as a
response variable.
\\

The usual formulation of the \textsc{Lasso} problem is the minimization
with respect to $\boldsymbol \beta$ of the quantity
\begin{equation}
  \label{eq:pbm_lasso}
  \frac  1 2  \sum_{i=1}^n  \left(y_i -\sum_{j=1}^p  x_{ij} \beta_j  \right)^2
  +\rho \|\boldsymbol\beta\|_{\ell_1},
\end{equation}
where   $(y_i)_{i=1,\dots,n}$   is    a   vector   of   response   and
$(x_{ij})_{i=1,\dots,n;j=1,\dots,p}$  a matrix  of predictors  such that 
$\sum_i  x_{ij} =0$,  with  no  loss of  generality.   
%When there  are several   uncorrelated   predictors,  
Using a coordinate-descent approach, we   simply   write  the   problem
\eqref{eq:pbm_lasso} in the form
\begin{equation*}
  \frac  1  2  \sum_{i=1}^n   \left(y_i  -\sum_{k\neq  j}  x_{ik}  \beta_k  -
    x_{ij}\beta_j \right)^2 +\rho \sum_{k\neq j} |\beta_k| +\rho |\beta_j|
\end{equation*}
and  minimizing this function  with respect  to $\beta_j$  will lead  to the
solution
\begin{equation*}
  \beta_j(\rho) =S\left(\sum_{i=1}^n x_{ij} (y_i-\tilde y_i^{(j)}),\rho\right) N_j^{-2}, 
\end{equation*}
where  $\tilde  y_i^{(j)}  =\sum_{k\neq  j}  x_{ik}\beta_k(\rho)$,  the normalizing term $N_j^2 $ satisfies $ N_j^2 =\sum_{i=1}^n   x_{ij}^2$ and   the   function    $S(x,\rho)   = \mathrm{sgn}(x)(|x|-\rho)_+$ is the soft-thresholding operator. 

This leads to  an iterative procedure, repeated on  each coordinate of
$\boldsymbol\beta$ until stabilization of  the full vector.  Note that
as   each   coordinate-wise   solution   is   unique,   results   from
\citet[Theorem 4.1]{2001_JOTA_Tseng} imply that the procedure converges.
\\

Now,  we   want  to   apply  this  approach   to  solve   the  problem
\eqref{eq:M_step_lasso}, which can be written
\begin{equation}\label{eq:lasso_friedman}
  \min_{\boldsymbol\beta}              \frac{1}{2}             \left\|
    \frac{1}{\sqrt{2}}\widehat{\boldsymbol\Sigma}_{11}^{1/2}\boldsymbol\beta
    -    \sqrt{2}\widehat{\boldsymbol\Sigma}_{11}^{-1/2}\mathbf{s}_{12}
  \right\|_2^2    +    \left\|    \mathbf{p}_{12}\star\boldsymbol\beta
  \right\|_{\ell_1}.
\end{equation}
From the  previous lines,  the  solution for  $j$th entry  of
$\boldsymbol\beta$ is
\begin{equation*}
  \beta_j(\mathbf{p}_{12}) = S\left(\sum_i (\widehat{\boldsymbol\Sigma}_{11}^{1/2})_{ij}
    \left( (\widehat{\boldsymbol\Sigma}_{11}^{-1/2}\mathbf{s}_{12})_{i}-\frac{1}{2}\sum_{k\neq
        j} 
      (\widehat{\boldsymbol\Sigma}_{11}^{1/2})_{ik}
      \beta_k(\mathbf{p}_{12})   \right),   (\mathbf{p}_{12})_j\right)
  N_j^{-2}.
\end{equation*}
Then, using the symmetry of the matrices, it is easy to see that
\begin{eqnarray*}
  & \sum_i (\widehat{\boldsymbol\Sigma}_{11}^{1/2})_{ij}  (\widehat{\boldsymbol\Sigma}_{11}^{-1/2} \mathbf{s}_{12})_i = \sum_\ell
  (\widehat{\boldsymbol\Sigma}_{11}^{1/2}\widehat{\boldsymbol\Sigma}_{11}^{-1/2})_{j\ell} (\mathbf{s}_{12})_{\ell} = (\mathbf{s}_{12})_{j} , \\
  &    \sum_i   (\widehat{\boldsymbol\Sigma}_{11}^{1/2})_{ij}\sum_{k\neq    j}   (\widehat{\boldsymbol\Sigma}_{11}^{1/2})_{ik}
  \beta_k(\mathbf{p}_{12}) = \sum_{k\neq j} (\widehat{\boldsymbol\Sigma}_{11})_{jk} \beta_k(\mathbf{p}_{12}) , \\
  & N_j^2 = \sum_i \left(\frac{(\widehat{\boldsymbol\Sigma}_{11}^{1/2})_{ij}}{\sqrt{2}}\right)^2 =  (\widehat{\boldsymbol\Sigma}_{11}/2)_{jj}.
\end{eqnarray*}
Finally,  the  solution  to  \eqref{eq:M_step_lasso}  is  computed  by
updating the $j$th coordinate of $\boldsymbol\beta$ via
\begin{equation*}
  \beta_j(\mathbf{p}_{12}) = 2 S\left((\mathbf{s}_{12})_j - \frac{1}{2} \sum_{k\neq j} (\widehat{\boldsymbol\Sigma}_{11})_{jk}
    \beta_k(\mathbf{p}_{12}) ; (\mathbf{p}_{12})_j \right) / (\widehat{\boldsymbol\Sigma}_{11})_{jj},
\end{equation*}
and permuting the rows of $\widehat{\boldsymbol\Sigma}$ until convergence.

\subsection{Reconstruction of the concentration matrix}\label{ap:inversion}

At  the  end  of  the  block-wise  resolution  algorithm,  a  solution
$\widehat{\boldsymbol\Sigma}  $  is available.   In  order to  recover
$\widehat{\mathbf{K}}$,    we    simply     use    the    fact    that
$\widehat{\boldsymbol\Sigma}\widehat{\mathbf{K}}  = I$.  Block-wisely,
we get
\begin{align*}
  \widehat{\mathbf{K}}_{12} & = -\widehat{\boldsymbol\Sigma}_{11}^{-1} \widehat{\boldsymbol\sigma}_{12} K_{22} = - K_{22} \widehat {\boldsymbol\beta}/2, \\
  \widehat{K}_{22}    &   =    1/(\widehat{\boldsymbol\sigma}_{12}   -
  \widehat{\boldsymbol\sigma}_{12}^\intercal
  \widehat{\boldsymbol\Sigma}_{11}^{-1}
  \widehat{\boldsymbol\sigma}_{12})                                   =
  1/(\widehat{\boldsymbol\sigma}_{12}                                 -
  \widehat{\boldsymbol\sigma}_{12}^\intercal                   \widehat
  {\boldsymbol\beta}/2),
\end{align*}
thanks   to   the   fact  that   $\widehat{\boldsymbol\sigma}_{12}   =
\widehat{\boldsymbol\Sigma}_{11}\widehat{\boldsymbol\beta}/2$.

To perform this  inversion, note that we need  to stock the successive
solutions $\widehat  {\boldsymbol\beta}$ of the  penalized regressions
along the algorithm.

\subsection{Pseudo-likelihood of a Gaussian vector}\label{ap:pseudo-lik}

It is well known that the  distribution of $X_i^k $ conditional on the
remaining variables  $X_{\backslash i}^k$ is  Gaussian with parameters
$(\mu_i^k,\sigma_i)$ given by
\begin{equation}
  \label{eq:mui_sigmai}
  \mu_i^k = {\boldsymbol \Sigma}_{i\backslash i} {\boldsymbol \Sigma}^{-1}_{\backslash i \backslash i} X_{\backslash i}^k, \qquad \sigma_i
  = {\Sigma}_{ii} - {\boldsymbol \Sigma}_{i \backslash i} {\boldsymbol \Sigma}^{-1}_{\backslash i \backslash i} {\boldsymbol \Sigma}^\intercal_{i \backslash i}. 
\end{equation}
Denoting  ${\boldsymbol  m}_i= (\mu_i^1,\dots,\mu_i^n)^\intercal$,  we
get
\begin{equation*}
 \log \tilde{\mathcal{L}}({\mathbf X};{\boldsymbol K}) = -\frac{n}{2} \sum_{i=1}^p \log \sigma_i
  - \sum_{i=1}^p \frac{1}{2\sigma_i}({\mathbf X}_i- {\boldsymbol m}_i)^\intercal ({\mathbf X}_i- {\boldsymbol m}_i) +c.
\end{equation*}
It  is easy to  see that  ${\boldsymbol m}_i^\intercal  = {\boldsymbol
  \Sigma}_{i  \backslash   i}{\boldsymbol  \Sigma}^{-1}_{\backslash  i
  \backslash i} {\mathbf X}_{\backslash i}^\intercal$. Then,
\begin{multline*}
 \log  \tilde{\mathcal{L}}({\mathbf X};{\boldsymbol K}) = -\frac{n}{2} \sum_{i=1}^p \log \sigma_i \\
  -   \sum_{i=1}^p   \frac{1}{2\sigma_i}({\mathbf X}_i^\intercal   {\mathbf X}_i   -   2
  {\mathbf X}_i^\intercal  {\mathbf X}_{\backslash  i}  {\boldsymbol \Sigma}^{-1}_{\backslash  i \backslash  i}  {\boldsymbol \Sigma}_{i  \backslash
    i}^\intercal  +  {\boldsymbol \Sigma}_{i \backslash  i}{\boldsymbol \Sigma}^{-1}_{\backslash  i  \backslash i}  {\mathbf X}_{\backslash
    i}^\intercal  {\mathbf X}_{\backslash  i}  {\boldsymbol \Sigma}^{-1}_{\backslash  i \backslash  i}  {\boldsymbol \Sigma}_{i  \backslash
    i}^\intercal) +c.
\end{multline*}
Note that we have $n^{-1} {\mathbf X}_i^\intercal  {\mathbf X}_i = S_{ii}$, as well as 
$n^{-1}  {\mathbf X}_i^\intercal  {\mathbf X}_{\backslash i}  =  \mathbf{S}_{i\backslash  i}$  and  $n^{-1}  {\mathbf X}_{\backslash   i}^\intercal {\mathbf X}_{\backslash i} = \mathbf{S}_{\backslash i\backslash i}$. Thus, 
\begin{multline}
  \label{eq:logL_vec3}
 \log \tilde{\mathcal{L}}({\mathbf X};{\boldsymbol K}) = - \frac{n}{2} \sum_{i=1}^p \log  \sigma_i \\
  -  n   \sum_{i=1}^p  \frac{1}{2\sigma_i}(S_{ii}  -  2   \mathbf{S}_{i  \backslash  i}
  {\boldsymbol \Sigma}^{-1}_{\backslash  i \backslash i}  {\boldsymbol \Sigma}_{i \backslash  i}^\intercal +  {\boldsymbol \Sigma}_{i \backslash
    i}{\boldsymbol \Sigma}^{-1}_{\backslash i  \backslash i} \mathbf{S}_{\backslash i  \backslash i} {\boldsymbol \Sigma}^{-1}_{\backslash  i \backslash i}
  {\boldsymbol \Sigma}_{i \backslash i}^\intercal) +c.
\end{multline}
Recalling  that  $\mathbf{K}=\mathbf{{\boldsymbol  \Sigma}}^{-1}$, and  by
reordering the  rows and columns of  the matrices, as well  as using a
block-wise notation, this becomes
\begin{equation*}
    \begin{bmatrix}
    { \Sigma}_{ii} & {\boldsymbol \Sigma}_{i \backslash i}\\
    {\boldsymbol \Sigma}_{i \backslash i}^\intercal & {\boldsymbol \Sigma}_{\backslash i\backslash i}\\
  \end{bmatrix}
  \times
  \begin{bmatrix}
   K_{ii} & \mathbf{K}_{i \backslash i}\\
    \mathbf{K}_{i \backslash i}^\intercal & \mathbf{K}_{\backslash i\backslash i}\\
  \end{bmatrix}
  =
  \begin{bmatrix}
    1 & 0\\
    0 & I_{p-1}\\
  \end{bmatrix},
\end{equation*}
where $I_{p-1}$ is the identity matrix with size $p-1$. In particular,
this leads to  the identity ${ \Sigma}_{ii} K_{ii}  = 1 - {\boldsymbol
  \Sigma}_{i      \backslash      i}     \mathbf{K}_{i      \backslash
  i}^\intercal$. Thus
\begin{equation}
  \label{eq:identity_1}
  { \Sigma}_{ii}   =   (1-   {\boldsymbol \Sigma}_{i   \backslash   i}   \mathbf{K}_{i   \backslash
    i}^\intercal)/K_{ii}. 
\end{equation}
In   the   same   way,   we   can  easily   get   that   ${\boldsymbol
  \Sigma}_{i\backslash   i}^\intercal    K_{ii}   =   -   {\boldsymbol
  \Sigma}_{\backslash   i    \backslash   i}\mathbf{K}_{i   \backslash
  i}^\intercal$ and
\begin{equation}
  \label{eq:identity_2}
  {\boldsymbol \Sigma}_{\backslash i \backslash  i}^{-1}{\boldsymbol \Sigma}_{i\backslash i}^\intercal = -\mathbf{K}_{i
  \backslash i}^\intercal/K_{ii}. 
\end{equation}
Using  identities   \eqref{eq:identity_1},  \eqref{eq:identity_2}  and
\eqref{eq:mui_sigmai}, we obtain
\begin{equation}
  \label{eq:sigmai}
  \sigma_i   =    (1-   {\boldsymbol \Sigma}_{i   \backslash    i}   \mathbf{K}_{i   \backslash    i}^\intercal)/K_{ii} + {\boldsymbol \Sigma}_{i\backslash
    i}   \mathbf{K}_{i  \backslash  i}^\intercal   /K_{ii}  =
  1/K_{ii} .
\end{equation}
Now,  coming back  to  \eqref{eq:logL_vec3} and  using the  identities
\eqref{eq:identity_1}, \eqref{eq:identity_2} and \eqref{eq:sigmai}, we
finally obtain the desired result
\begin{equation*}
 \log \tilde{\mathcal{L}}(\mathbf X;\mathbf{K}) =
  \frac{n}{2} \sum_{i=1}^p \log K_{ii} - n \sum_{i=1}^p \left(
  \frac{K_{ii}}{2} S_{ii} + \mathbf{S}_{i \backslash i} \mathbf{K}_{i \backslash i} +
  \frac{1}{2K_{ii}}\mathbf{K}_{i  \backslash  i}\mathbf{S}_{\backslash  i  \backslash  i}
  \mathbf{K}_{i \backslash i}^\intercal\right) +c.
\end{equation*}

\subsection{Penalization upper bound \label{sec:bound}}

The  following  lemma  states  that  if  the  penalization  parameters
$\lambda_{q\ell}^{-1}$  and $\lambda_0^{-1}$  are chosen  large enough
(according to the observations), then the penalized estimator obtained
from the \textsc{Lasso}-like iteration step has null entries.
 
\begin{lemma}
  If for any 
  $i,j \in \mathcal{P}$ we have
  \begin{equation}\label{eq:lambda_large}
    \sum_{q,\ell} \frac{Z_{iq}Z_{j\ell}}{\lambda_{q\ell}} \geq \frac n
    2 |S_{ij}| , \text{ when }   i\neq j \quad \text{ and } \quad  \frac 1
    {\lambda_0} \geq  \frac n  2 |S_{ii}| ,
  \end{equation}
  then the solution $\widehat {\boldsymbol \Sigma}= \widehat {\mathbf
  K}^{-1}$ of problem \eqref{eq:M_step} satisfies
$  \widehat {\mathbf K}^{-1}= 0$ .
\end{lemma}

\begin{proof}
  The sub-gradient equation arising from \eqref{eq:M_step} gives
\begin{multline*}
  \forall    i\neq    j,    \quad    \frac    n    2    \left(\widehat
    K_{ij}^{-1}-S_{ij}\right) -\left(\sum_{q,\ell}
    \frac{Z_{iq}Z_{j\ell}}{\lambda_{q\ell}} \right)\nu_{ij}=0 \\
  \text{and  }   \forall  i  \in   \mathcal{P}  ,  \quad  \frac   n  2
  \left(\widehat              K_{ii}^{-1}-S_{ii}\right)              -
  \frac{1}{\lambda_{0}}\nu_{ii}=0,
\end{multline*}
where    $\nu_{ij}\in   \textrm{sgn}(\widehat   K_{ij})$    and   thus
$\nu_{ij}\in [-1,1]$. In particular, we have
\begin{equation*}
  \forall i\neq j, \quad  \frac n 2 \left|\widehat
    K_{ij}^{-1}-S_{ij}\right| \le \left(\sum_{q,\ell}
    \frac{Z_{iq}Z_{j\ell}}{\lambda_{q\ell}} \right) 
  \ \text{and} \ \forall i \in \mathcal{P}, \ \frac n 2 \left|\widehat
    K_{ii}^{-1}-S_{ij}\right| \le  \frac{1}{\lambda_{0}}  .
\end{equation*}
Now, if the set of penalty parameters satisfies the constraint
\eqref{eq:lambda_large}, then the matrix $\mathbf K^{-1}=0$  satisfies the
sub-gradient equation. Thus, the conclusion comes from uniqueness of
the solution to \eqref{eq:M_step}.

\end{proof}

%%% Local Variables:
%%% mode: latex
%%% TeX-master: "CJC.tex"
%%% End: